\documentclass[10pt,reqno]{article}
\usepackage{fullpage}

\usepackage[unicode=true]{hyperref}
\usepackage{amsmath}
\usepackage{cite}
\usepackage{amsfonts}
\usepackage{amssymb}
\usepackage{amsthm}
\usepackage{authblk}
\usepackage{float}
\usepackage{booktabs}
\usepackage{geometry}

\usepackage[ruled,vlined]{algorithm2e}

\usepackage[pdftex]{color,graphicx}
\usepackage{mypack} 
\usepackage{adjustbox}
\usepackage{soul}
\usepackage{comment}
\usepackage{tikz}
\usetikzlibrary{decorations.pathreplacing,angles,quotes}
\usepackage{bbold}

\usepackage{diagbox}

\usetikzlibrary{matrix,arrows,decorations.pathmorphing}
\usepackage{tikz-cd}
\usetikzlibrary{arrows} 

\usetikzlibrary{decorations.markings}
 \usepackage{caption}
\usepackage{subcaption}
 
 \usepackage{pdfpages}

\usepackage{blkarray}

\usepackage{centernot}
\usepackage{enumerate}
\usepackage{enumitem}

\usepackage{mathtools}
\usepackage{stmaryrd}

\newcommand{\dashedentry}[2]{%
  \makebox[#1]{\hrulefill}\quad #2 \quad \makebox[#1]{\hrulefill}%
}

\definecolor{bluegray}{rgb}{0.4, 0.6, 0.8}
\definecolor{turquoise}{rgb}{0.2, 0.7, 0.6}
\definecolor{darkgreen}{rgb}{0.4, 0.6, 0.3}
\definecolor{pink}{rgb}{1, 0.65, 0.8}

\usepackage{todonotes}
\newcounter{commcount}
\usepackage{xcolor}

\definecolor{purpleish}{rgb}{0.5, 0.0, 0.5}
\definecolor{pinkish}{rgb}{0.8, 0.0, 0.4}

\usepackage{soul}  

\title{Phase space tableau simulation for quantum computation}

\author[1]{Selman Ipek}
\author[2]{Atak Talay Yucel}
\author[1]{Farzad Shahi}
\author[3]{Cagdas Ozdemir}
\author[1]{Cihan Okay\footnote{cihan.okay@bilkent.edu.tr}} 
\affil[1]{{\small{Department of Mathematics, Bilkent University, Ankara, Turkey}}}
\affil[2]{{\small{Department of Computer Engineering, Bilkent University, Ankara, Turkey}}}
\affil[3]{{\small{Department of Physics, Bilkent University, Ankara, Turkey}}}

\begin{document}
  \maketitle  
  
\begin{abstract}
{We introduce a novel tableau-based classical simulation method for quantum computation, formulated within the phase space framework of the extended stabilizer theory of closed non-contextual operators. This method} enables the efficient classical simulation of a broader class of quantum circuits beyond the stabilizer formalism. {We implement the simulator and benchmark its performance on basic quantum algorithms, including the hidden shift and Deutsch--Jozsa algorithms.}
\end{abstract}

 	\tableofcontents

\section{{Introduction}}

Classical simulation---the task of performing quantum algorithms using classical computation---plays a central role in understanding the computational power of quantum systems relative to their classical counterparts. Although simulating universal quantum circuits is widely believed to require exponential classical resources, various methods have been developed that perform efficiently in specific regimes. {Our work is motivated by a recent line of research on polyhedral simulators, which provide a geometric framework for systematically extending the boundaries of efficient classical simulation, offering new insights into the structure underlying quantum advantage.}

A  
{well-known}
tractable class of quantum operations is {described} by \textit{stabilizer theory}, which focuses on stabilizer input states, Clifford unitaries, and measurements of Pauli operators—optionally involving classical feed-forward and randomness. Among classical simulation techniques, stabilizer-based methods stand out for efficiently encoding certain quantum circuits, called stabilizer circuits, which can exhibit significant entanglement yet remain simulable in resources scaling only polynomially with $n$, the number of qubits. Central to this efficiency is the Gottesman–Knill (GK) theorem \cite{gottesman1999group22}, which guarantees that stabilizer circuits can be efficiently classically simulated.
Nonetheless, stabilizer circuits alone are not universal for quantum computation. Universality can be achieved by supplementing them with \textit{non-stabilizer} resources. In the model of quantum computing with magic states (QCM), non-Clifford operations are realized via the injection of specially prepared \textit{magic} states~\cite{magic}. 
{Alongside stabilizer rank methods~\cite{bravyi2016trading,bravyi2016improved,bravyi2019simulation}, approaches based on sampling from quasi-probability distributions have provided a concrete framework for systematically studying the classical simulation of QCM and its complexity. A prominent example is the work of Howard and Campbell~\cite{howard2017application}, where quantum states are represented as quasi-probabilistic mixtures of stabilizer states.}
Robustness of magic, the associated magic monotone, characterizes the minimal amount of negativity needed to represent a magic state and, in turn, controls the cost of classical simulation.
{The a}pplication of Howard and Campbell's approach relies on the efficient classical simulation of stabilizer circuits, which can be implemented directly using a tableau 
\cite{nielsen2010quantum} or its improved version introduced by Aaronson and Gottesman (AG) \cite{aaronson2004improved}. 
~More recently, Seddon et al. \cite{seddon2021quantifying} introduced a quasi-probability representation based on the notion of stabilizer dyads, whose corresponding dyadic frame simulator makes use of the phase sensitive tableau method of Bravyi et al. \cite{bravyi2016trading}, refining the tableau of \cite{aaronson2004improved}.

An entirely distinct quasi-probability representation, known as the \emph{phase space representation}, was proposed by Raussendorf et al.~\cite{raussendorf2020phase} as a {unified extension of} phase space methods based on Wigner functions, previously developed for qudits of odd-prime local dimension~\cite{veitch2012negative} and extended to qubits~\cite{raussendorf2017contextuality}.
Key to this approach are the so-called \emph{closed non-contextual (CNC) {operators}}, objects that {constitute the phase space and} extend  stabilizer states. 
{F}or qudits, it is possible to attribute to 
{every} 
Pauli observable a deterministic outcome, this is not possible for qubits due to the emergence of state-independent contextuality, a phenomenon  
{manifested}
in the Mermin square paradox \cite{mermin1993hidden}.%
~

In this work, we introduce a novel tableau method for the phase space simulation, called the {\textit{phase space tableau} or \textit{CNC tableau}}, 
enabling the efficient classical simulation of broader classes of quantum circuits than stabilizer circuits.  
The complexity of our simulation scales as that of the {AG} simulation. 
An implementation of our algorithm can be found in \cite{CNCSim}.

\begin{thm}\label{thm:cnc-tableau}
The following statements are valid for the CNC tableau:
\begin{enumerate}
\item A CNC tableau can be encoded using $O(n^{2})$ bits of memory.

\item The update of a CNC tableau under:
\begin{itemize}
\item Clifford unitaries can be performed in $O(n)$ time,
\item Pauli measurements can be performed in {$O(n^{2})$ time}.
\end{itemize}
\end{enumerate}
\end{thm}%

Our phase space simulator improves upon the {stabilizer} simulator by incorporating more structure than the standard stabilizer and destabilizer components in the tableau. Specifically, the {standard} qubit phase space---defined as Pauli operators modulo scalars---is decomposed into a stabilizer part, a destabilizer part, and an additional Jordan--Wigner component, {a special type of fermionic transformation}.
Under a Pauli measurement, each of the three components in our decomposition is updated according to a rule that refines the standard stabilizer update procedure. In {AG}, the improved version of the GK simulator~\cite{aaronson2004improved}, there are two update cases depending on whether the measured Pauli operator lies inside or outside the stabilizer group of the state. Similarly, the original CNC update rules of~\cite{raussendorf2020phase} also involve two cases. In contrast, our update rules are divided into four cases {(Case I–IV)}, more closely mirroring the three cases considered in~\cite{aaronson2004improved} for mixed stabilizer states. 
{Moreover, when used as a quasi-probability simulator, our method achieves improved efficiency over stabilizer-based approaches, as its simulation cost is governed by phase space robustness---a quantity known to be upper bounded by the robustness of magic~\cite{raussendorf2020phase}.
}

We organize the paper as follows. Section~\ref{sec:background} introduces the necessary background and notation. In Section~\ref{sec:jw-section}, we develop the {fundamental} notions of JW sets, bases, and decompositions. Section~\ref{sec:the update rules} applies these notions to analyze the update rules, and Section~\ref{sec:tableau representation} presents the construction of the CNC tableau. {Finally, Section~\ref{sec:phase space tableau simulation} is devoted to complexity analysis and benchmarking on basic quantum algorithms.}

{
\paragraph{Acknowledgements.}
This work is supported by the Digital Horizon Europe project FoQaCiA, GA no. 101070558. The last author also acknowledges support from the US
Air Force Office of Scientific Research under award number FA9550-24-1-0257.  
}  
  
\section{{The phase space}}\label{sec:background}

In this section, we review stabilizer theory {\cite{nielsen2010quantum}} and {its extension based on closed non-contextual (CNC) operators, as introduced in~\cite{raussendorf2020phase}.}

\subsection{Stabilizer theory}

We consider a system of $n$ qubits modeled by a Hilbert space $\hH = (\CC^{2})^{\otimes n}$. Given a bit string $a = (a_{X},a_{Z})\in E_{n}$, where $E_{n} := \ZZ_{2}^{2n}$, a Pauli operator is a 
{Hermitian operator}
of the form
\begin{eqnarray}
T_{a} = i^{\phi(a)}X^{a_{X}}Z^{a_{Z}},
\end{eqnarray}
where $\phi:E_{n}\to \ZZ_{4}$ is a phase function that is conventionally chosen to be $\phi(a) = a_{X}\cdot a_{Z}$. Here, we have adopted the notation $X^{a_{X}} = X^{(a_{X})_{1}}\otimes \cdots \otimes X^{(a_{X})_{n}}$ with a similar notation for $Z^{a_{Z}}$. 


{The vector space $E_n$ has a symplectic {form} defined by} 
$\left [a,b \right ] = a_{X}\cdot b_{Z}+a_{Z}\cdot b_{X}$ where the dot product is the usual 
{dot}
product {over $\ZZ_2$}.   
Pauli operators $T_{a}$ and $T_{b}$ then satisfy $T_{a}T_{b} = (-1)^{[a,b]}T_{b}T_{a}$ and thus commute if and only if the corresponding symplectic {form} $[a,b] = 0$ and anti-commute otherwise. 
The product of  
{commuting}
Pauli operators satisfies  
$T_{a}T_{b} = (-1)^{\beta(a,b)}T_{a+b}$; {see for example}  
\cite{okay2017topological}.
Using our phase convention $\beta$ has an explicit form given by%
{%
\begin{eqnarray}
\beta(a,b) = {\frac{\phi(a) + \phi(b) + 2(a_z \cdot b_x) - \phi(a \oplus b)}{2}},
\label{eq:beta}
\end{eqnarray}
{where $\oplus$ denotes addition in $\ZZ_{2}$ and other operations are over $\ZZ_4$.}
}

The Pauli group $P_{n}$ is given by the set of operators $\{\pm 1, \pm i\}\times \{T_{a}~:~a\in E_{n}\}$ with matrix multiplication as the group operation. Its quotient with respect to its center $\Span{iI}$ is denoted $\bar P_n$. {This group is isomorphic to $E_n$}. The $n$-qubit Clifford group is the normalizer of the Pauli group with phases quotiented out
\begin{eqnarray}
\Cl_{n} := \left \{U\in \uU(\hH)\,:\,UP_{n}U^{\dagger} = P_{n} \right \}/\Span{e^{i\theta}I\,:\,\theta\in \RR}.
\end{eqnarray}
We can further identify the quotient group $\Cl_{n}/\bar P_{n}$ with the symplectic group $\Sp_{2n}(\ZZ_{2})$, the matrix group of invertible linear maps $E_{n}\to E_{n}$ that preserve the symplectic {form}.
{That is, there is a surjective group homomorphism
\[\varphi:\Cl_{n}\to {\text{Sp}_{2n}(\ZZ_{2})}.\]}
In fact, any Clifford unitary $U_{{S}}$ acts by conjugation on Pauli operators according to 
\[
U_{S}T_{a}U_{S}^{\dagger} = \pm T_{S(a)}
\]
where $S$ is symplectic map.

{A} subspace $I\subset E_{n}$ is called \textit{isotropic} if for every $a,b\in I$ we have that $[a,b] = 0 $. 
 The orthogonal complement $I^{\perp}$ of an isotropic subspace $I$ consists of all elements that commute with $I$.  
Note that $I\subset I^{\perp}$. Moreover, if $\text{dim}(I)=n-m$ then $\text{dim}(I^{\perp}) = n+m$ {and} $\text{dim}(I) + \text{dim}(I^{\perp}) = 2n = \text{dim}(E_{n})$. An isotropic subspace is called {\it maximal} if $I = I^{\perp}$, i.e. it is self-dual. 
A stabilizer code is a projector%
$$\Pi_{I}^{s} = \frac{1}{|I|}\sum_{a\in I}(-1)^{s(a)}T_{a}$$%
where $I$ is an isotropic and $s:I\to \ZZ_{2}$ is a value assignment. When $I$ is maximal then $\Pi_{I}^{s}$ is precisely a stabilizer state. {At the other extreme, when $ I = \Span{a} $, we denote the corresponding projector as $\Pi_a^s$, where $s$ represents the value in $\mathbb{Z}_2$ assigned to $a$.}
{By Witt's lemma the symplectic group acts on the set of isotropic subspaces of fixed dimension in a transitive way. Similarly, the action of the Clifford group on the stabilizer projectors of fixed rank is transitive.}

\subsection{Closed non-contextual operators} 
 
Closed non-contextual operators {provide} a natural extension of the stabilizer theory.  
 
{
\Def{\label{def:closed non-contextual set}
A subset $\Omega\subset E_{n}$ is called \textit{closed} if, for every pair $a,b\in\Omega$ satisfying $[a,b]=0$, it follows that $a+b\in \Omega$. For a closed set $\Omega$, a function $\gamma:\Omega\to \ZZ_{2}$ is called a \textit{value assignment} if, for all commuting $a,b\in \Omega$, the relation
\begin{eqnarray}
\gamma(a+b) = \gamma(a)+\gamma(b)+\beta(a,b),\label{eq:value-assignment}
\end{eqnarray}
holds, where $\beta(a,b)$ is defined in Eq.~(\ref{eq:beta}). 
A closed subset $\Omega \subset E_{n}$ that admits such a value assignment is called a closed non-contextual (CNC) set.%
}
}%
 
The structure of CNC sets was extensively studied in \cite{raussendorf2020phase}; see also \cite{kirby2019contextuality}. 
We call a CNC set \textit{maximal} if it is not contained in any other {CNC set}.%
~Given a CNC set $\Omega$ and a value assignment $\gamma:\Omega\to\ZZ_{2}$ (sometimes called a \textit{CNC pair}) we define a unit trace Hermitian operator
\begin{eqnarray}\label{eq:CNC operator}
A_{\Omega}^{\gamma} := \frac{1}{2^{n}}\sum_{a\in\Omega}(-1)^{\gamma(a)}T_{a}.
\end{eqnarray}
Operators of this form are called {\it CNC operators}. {Stabilizer states are particular examples of CNC operators.}

\subsection{Classical simulation}
 
{CNC operators} span the space of unit trace Hermitian operators. In particular, any quantum state $\rho$ can be expressed by an affine mixture of maximal CNC operators. Thus these operators form a Clifford covariant finite state space that can be used for a sampling based classical simulation algorithm \cite{raussendorf2020phase}. 
{This algorithm belongs to a class of sampling algorithm on the set of vertices of an operator-theoretic polytope. 
Let $\SP_n$ denote the stabilizer polytope defined as the convex hull of $n$-qubit stabilizer states. The dual of this polytope, denoted by $\Lambda_n$, provides a probabilistic algorithm for universal quantum computation \cite{zurel2020hidden}. Furthermore, it was also shown for all $n$ that CNC operators with maximal $\Omega$ are vertices of the $\Lambda$ polytopes.
Let $\CP_n$ denote the CNC polytope defined as the convex hull of the CNC operators.} 
It is not difficult to show ({see Corollary~\ref{cor:maximality}}) that stabilizer states are not maximal CNC operators and thus they lie in {$\CP_{n}$}, so that 
{\[
\SP_n \subset \CP_n \subset \Lambda_n .
\]}

Classical simulation based on the CNC state space extends and improves the classical simulation based on stabilizer states.%
~As demonstrated in \cite{raussendorf2020phase}, the update of CNC operators under Clifford unitaries and Pauli measurements can be efficiently simulated on a classical computer.%
~{In \cite{raussendorf2020phase} they demonstrate this using}%
~the characterization of CNC sets (described in more detail below) as a union of isotropic subspaces, each of which can essentially {be} updated via a stabilizer tableau.%
~{The goal of the present work is to refine the analysis of~\cite{raussendorf2020phase}, making the implementation and computational efficiency of the update rules more explicit.}

The amount of negativity required to represent any quantum state $\rho$ as an affine mixture of maximal CNC operators is less than when using stabilizer states. This is quantified by the \textit{robustness} measure
\begin{eqnarray}
\mathfrak{R}(\rho) := \min_{{W}\,:\,\rho = \Span{{W},A}} \norm{{W}}_{1}
\end{eqnarray}
where $\rho = \sum_{\alpha}{W_{\rho}(\alpha)}A_{\alpha}$ is an affine mixture and the sum is over CNC pairs labeled by $\alpha$. The robustness is a monotone under Clifford unitaries and Pauli measurements \cite{raussendorf2020phase} and is upper bounded by $\mathfrak{R}_{S}$, the analogous measure using stabilizer states \cite{howard2017application}. By the arguments in~\cite{pashayan2015estimating}, classical simulation using $\CP_n$ is therefore more efficient than simulation based on stabilizer states. {This gain in efficiency motivates the implementation of our algorithm.}

\section{{Structure of the phase space}}\label{sec:jw-section}

{In this section, we introduce the symplectic preliminaries required for constructing the tableau algorithm. The key concept is a new decomposition of the symplectic vector space $E_n = \mathbb{Z}_2^{2n}$ into stabilizer, destabilizer, and Jordan--Wigner components, which we refer to as {a} Jordan--Wigner decomposition. This decomposition allows us to encode additional information in our {tableau}, extending beyond stabilizer {tableau}.
}

\subsection{{Jordan-Wigner decomposition}}

A set of elements 
\[
\mathcal{S} = \set{a_1, \cdots, a_N} \subset E_n
\]
is called an {\it anti-commuting set} if {its elements pairwise anti-commute}, i.e., $[a_i,a_j]=1$ for all $i\neq j$. An anti-commuting set is called a maximal anti-commuting set if it is not strictly contained in any other anti-commuting set. 

{We begin with some basic observations on maximal anti-commuting sets; see also \cite{sarkar2019sets}.} 

\Lem{\label{lem:maximal_anticommutings_are_odd}
For an anticommuting set $\mathcal{S}$ with $N=|\mathcal{S}|$ even, the sum
$
b= \sum_{i=1}^{N} a_i
$  
is non-zero and adding this element to the set results in a maximal anti-commuting set.   
} 
\Proof{{If $N=0$, then $\mathcal{S}=\set{0}$ is clearly maximal. So, assume $N > 0$.} {If $b=0$ then $a_N=\sum_{i=1}^{N-1} a_i$ and $[a_{N},a_{i}]={\sum_{j=1}^{N-1}[a_j, a_i]}=0$ for $i=1,\cdots,N-1$. Therefore $b$ is non-zero. Adding this element to the set gives a maximal anti-commuting set since any other element that anti-commutes with all the $a_i$'s commutes with $b$.}
}

For an anti-commuting set $\mathcal{S}$ and an element $c\in E_n$, we define
\begin{equation}\label{eq:A and C}
{\aA_c} = \set{a_i \mid [a_i, c] = 1, i = 1, \cdots, N} \quad \text{and} \quad {\cC_c} = \set{a_i \mid [a_i, c] = 0, i=1, \cdots, N}.
\end{equation}%
{
Note that the definition of these sets also depends on $\mathcal{S}$, which is omitted from notation for clarity.}

\Pro{\label{pro:maximal_anticommuting} 
An anti-commuting set $\mathcal{S}$ is maximal if and only if $\sum_{i=1}^{N} a_i = 0$. Furthermore, in this case, $N$ must be odd.}
\Proof{ 
First, assume that $\mathcal{S}$ is maximal {and $b = \sum_{i=1}^{N}a_i$}. Lemma \ref{lem:maximal_anticommutings_are_odd} implies that $N$ is odd. 
Let $c \in E_n$ {be} such that $[b, c] = 1$.
Then,
\[
[b,c]=\sum_{i=1}^{N}[a_{i},c]=\sum_{a\in \aA_c}[a,{c}]=1,
\]
which implies that $|\aA_c|$ is odd and $|\cC_c|$ is even since $N$ is odd. The element 
$
d=c+\sum_{a\in \cC_c} a 
$
satisfies $[d,a_i]=1$ for all $i=1,\cdots,N$ and $d\not\in \mathcal{S}$. This contradicts to the maximally of $\mathcal{S}$. {Hence, no such $c$ exists and this means $b = 0$.}
For {the converse}, assume that $b = \sum_{i=1}^{N} a_i = 0$. {By Lemma \ref{lem:maximal_anticommutings_are_odd}, $N$ must be odd. Since we have $a_N = \sum_{i=1}^{N-1}a_i$ and $|\mathcal{S} {\setminus} \set{a_{N}}|$ is even,
Lemma~\ref{lem:maximal_anticommutings_are_odd} implies that $\mathcal{S}$ is maximal.} 
}

Notice that {$\mathcal{S}$} is linearly dependent if and only if there exists a nonempty $I \subset \set{1, \cdots, N}$ such that $\sum_{i \in I} a_i = 0$. 
  The latter condition is equivalent to $\set{a_i}_{i \in I}$ being a maximal anti-commuting set by Proposition~\ref{pro:maximal_anticommuting}.
{Therefore, we have the following:}
    \begin{itemize}
      \item If $\mathcal{S}$ is maximal, then any subset of $\mathcal{S}$ with size $N-1$ is linearly independent.
      \item If $\mathcal{S}$ is not maximal, then $\mathcal{S}$ is linearly independent.
    \end{itemize}
As a consequence, we can see that the maximum size of an anti-commuting set is $2n+1$. These max-sized anti-commuting sets are particularly important for our discussion.



\Def{\label{def:JW elements}
An anti-commuting set of $2n+1$ elements, $\{a_1, \dots, a_{2n+1}\}$, is called a \textit{Jordan-Wigner (JW) set}, and its elements are referred to as JW elements.
{Any $2n$ JW elements form a basis of $E_n$ that we call a \textit{JW basis}.}
}

We typically choose the first $2n$ JW elements as the basis so that the remaining element $a_{2n+1}$ is the sum of the JW basis elements.

\Lem{\label{lem:extending_nonmaximal_anticommuting_to_jw}
  Any anti-commuting set 
 that is not maximal 
 can be extended to a JW set.
}



\Proof{{Let $\mathcal{S}_N= \set{a_1, \cdots, a_N}$ be a {non-maximal} anti-commuting set.}
If $N = 2n$ then the result directly follows from Lemma~\ref{lem:maximal_anticommutings_are_odd}.  
Therefore, we are done if we can show that $\mathcal{S}_{N}$ can be extended to an anti-commuting set $\mathcal{S}_{N+1} = \mathcal{S}_{N}\cup \set{a_{N+1}}$ that is not maximal when $N < 2n$. 
  If $N$ is odd, there exists $a_{N+1} \in E_n$ such that $[{a_{N+1}}, a_i] = 1$ for all $i \in \set{1, \cdots, N}$ since $\mathcal{S}_N$ is not maximal. Since $N+1$ is even, the extended set $\mathcal{S}_{N+1}$ is not maximal by Lemma~\ref{lem:maximal_anticommutings_are_odd}.
Now if $N$ is even, then we can find an element $c \in E_n \setminus \Span{a_1, \cdots, a_N}$ since $N < 2n$.  
{W}e can construct
  \[ a_{N+1} = \begin{cases}
    c + \sum_{a \in \aA_c} a & \text{if } |\cC_c| \text{ is odd} \\
    c + \sum_{a \in \cC_c} a & \text{if } |\cC_c| \text{ is nonzero and even}\\
    c & \text{if } |\cC_c| = 0.
  \end{cases}
  \]
  It is easy to see that $[a_{N+1}, a_i] = 1$ for all $i \in \set{1, \cdots, N}$. So, $\set{a_1, \cdots, a_N, a_{N+1}}$ is an anti-commuting set.
  Moreover, $\sum_{i=1}^{N+1} a_i \neq 0$ since {otherwise} it would imply that $c \in \Span{a_1, \cdots, a_N}$. Hence, $\set{a_1, \cdots, a_N, a_{N+1}}$ is not maximal.
}

{
\Def{\label{def:JW symplectic transformations}
The Jordan–Wigner (JW) transform of a symplectic basis $\{e_i, f_i\}_{i=1}^n$ is given by the anticommuting set:
\[
a_1 = e_1, \quad a_2 = f_1,
\]
and for $i = 2, \ldots, n$,
\begin{equation} \label{eq:symplectic_to_jw}
  a_{2i-1} = e_i + \sum_{k=1}^{i-1} (e_k + f_k), \quad
  a_{2i} = f_i + \sum_{k=1}^{i-1} (e_k + f_k).
\end{equation}

Conversely, the inverse JW transform of an anticommuting set $\{a_1, \ldots, a_{2n}\}$ yields the symplectic basis $\{e_i, f_i\}_{i=1}^n$, defined by
\[
e_1 = a_1, \quad f_1 = a_2,
\]
and for $i = 2, \ldots, n$,
\begin{equation} \label{eq:jw_to_symplectic}
  e_i = a_{2i-1} + \sum_{k=1}^{2i-2} a_k, \quad
  f_i = a_{2i} + \sum_{k=1}^{2i-2} a_k.
\end{equation}
}}

{T}hese transformations are inverses of each other. So, $\Span{a_1, \cdots, a_{2n}} = \Span{e_1, f_1, \cdots, e_n, f_n}$ in both cases. 
The JW elements associated to the canonical symplectic basis $\set{z_i,x_i}_{i=1}^{2n}$ are called the {\it canonical JW elements}.
{Moreover, 
\begin{equation}\label{eq:W symplectic 2k}
W = \Span{a_1,\cdots,a_{2k}}
\end{equation}
is a symplectic vector space of dimension $2k$.
}

\Lem{\label{lem:JW basis}
  Any $b\in E_n$ can be expressed in the JW basis as
  \begin{equation}\label{eq:a_k basis}
    b=\sum_{k=1}^{2n} \nu_k a_k
  \end{equation}
  where $\nu_k = [a_{2n+1}+a_k,b]$.
}
\Proof{
  \begin{align*}
    [a_{2n+1} + a_k, b] &= [a_1, b] + \cdots + [a_{k-1}, b] + [a_{k+1}, b] + \cdots + [a_{2n}, b]\\
    &= \nu_1 + \cdots + \nu_{k-1} + \nu_{k+1} + \cdots + \nu_{2n} + (2n-1)\sum_{i=1}^{2n} \nu_i\\
    &= \nu_1 + \cdots + \nu_{k-1} + \nu_{k+1} + \cdots + \nu_{2n} + \sum_{i=1}^{2n} \nu_i\\
    &= \nu_k.
  \end{align*}
  {where we used $a_{2n+1} = \sum_{k=1}^{2n} a_k$ and $[a_k, b] = \nu_k + \sum_{i=1}^{2n} \nu_i$ in the second equality.}
}

It is convenient to write $\nu_{{2n+1}}=[a_{{2n+1}},b]$, so that $\nu_k=\nu_{2n+1}+[a_k,b]$. 

\Cor{\label{cor:commuting vs anticommuting}
{For $b\in E_n$, let $\aA_b$ and $\cC_b$ be the sets defined in Eq.~(\ref{eq:A and C}) with respect to a JW set. 
Then,}
\[
b = \sum_{a\in \aA_b} a = \sum_{a'\in \cC_b} a'
\] 
and $|\aA_b|$ is even.
}
\Proof{

Using Eq.~(\ref{eq:a_k basis}), we have
\begin{align*}
b = \sum_{k=1}^{2n} \nu_k a_k = \nu_{2n+1} \sum_{k=1}^{2n} a_k + \sum_{k=1}^{2n} [a_k,b] a_k { = \nu_{2n+1} a_{2n+1}+ \sum_{a \in \bar{\aA_b}} a~.}
\end{align*}
{where $\bar{\aA_b} = \aA_b \setminus \set{a_{2n+1}}$}.
Now if $a_{2n+1} \notin \aA_b$, then $\nu_{2n+1} = 0$ and {$\aA_b = \bar{\aA_b}$. So, }$b = \sum_{a \in \aA_b} a$. If $a_{2n+1} \in \aA_b$,
then $\nu_{2n+1} = 1$ and {$\aA_b = \bar{\aA_b} \cup \set{a_{2n+1}}$ which gives}
\begin{align*}
  b = a_{2n+1} + \sum_{a \in {\bar{\aA_b}}} a = \sum_{a \in \aA_b} a.
\end{align*}
Since $\sum_{i=1}^{2n+1} a_i = 0$, we have $\sum_{a \in \aA_b} a + \sum_{a' \in \cC_b} a'= 0$ and this implies $ b = \sum_{a' \in \cC_b} a'$.
Lastly, if $|\bar{\aA_b}|$ is odd, then $a_{2n+1} \in \aA_b$ since $a_{2n+1} = \sum_{i = 1}^{2n} a_i$. If $|\bar{\aA_b}|$ is even, then $a_{2n+1} \notin \aA_b$. 
Hence, $|\aA_b|$ is even.
}

{Now, we come to the main definition of this section.}

\Def{\label{def:JW-decomposition}
A {\it Jordan--Wigner decomposition} for $E_n$ of type $m$ consists of
\begin{itemize}
\item an isotropic subspace ${I_{n-m}}\subset E_n$  of dimension {$n-m$},
\item a subspace ${W_{2m}}\subset I^\perp$ {of dimension $2m$} such that $W_{2m}\cap I_{n-m}=0$ and {$W_{2m}{\oplus}I_{n-m}=I_{n-m}^\perp$},
\item an isotopic subspace ${I_{n-m}'}\subset E_n$ of dimension {$n-m$} such that $I_{n-m}'\cap I_{n-m}^\perp=0$ {and $I_{n-m}'{\oplus}W_{2m}=(I_{n-m}')^\perp$}.
\end{itemize}
The {subspaces in the} pair $(I_{n-m},I_{n-m}')$ will be referred to as  {\it conjugate} isotropic subspaces.
} 

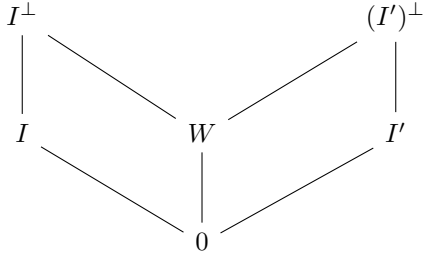
\begin{figure}[htbp]
    \centering
    \[
\begin{tikzcd}[column sep=huge, row sep=large]
I^\perp \arrow[d,no head] \arrow[dr,no head] && (I')^\perp \arrow[d,no head] \arrow[dl,no head] \\
I \arrow[dr,no head] & W \arrow[d,no head] & I' \arrow[dl,no head] \\
&0&    
\end{tikzcd}
\]
    \caption{{The intersections and sums of the subspaces in a JW decomposition: {$W\cap I=0$, $W\oplus I=I^\perp$, $W\cap I'=0$, and $W\oplus I'=(I')^\perp$}}}
    \label{fig:JW decomposition} 
\end{figure}

We will write 
\[
E_n = I_{n-m} \oplus W_{2m} \oplus I_{n-m}'
\]
for a decomposition of type $m$. {Subscripts indicating the dimensions will usually be omitted to simplify the notation. The relationship between these components is depicted in Fig.~(\ref{fig:JW decomposition}). Note that we have 
\[I^\perp \cap (I')^\perp =W.\]}

{

\Lem{\label{lem:conjugate basis}
Given conjugate isotropic {subspaces} $I$ and $I'$, we can choose a basis $\set{e_1,\cdots,e_{n-m}}$ and $\set{f_1,\cdots,f_{n-m}}$, respectively, such that
\[
[e_i,f_j]=\delta_{ij}
\]
where $1\leq i,j\leq n-m$.
}
\Proof{
It is a well-known fact in stabilizer theory that for $b\notin I^\perp$, we can make a change of basis so that only one of the basis vectors anti-commute with $b$, and the rest commute \cite{nielsen2010quantum}. 
Let $\lambda$ be the minimum of the set $\set{1\leq i\leq n-m \mid [e_i,b]=1}$. Then,
we can write 
\begin{align}\label{eq:modified basis}
I &= \Span{\bar{e}_{1}, \cdots, \bar{e}_{n-m}}\notag\\
\bar e_i &= \begin{cases}
  e_i + e_\lambda & \text{if } [e_i,b]=1 \text{ and } i \neq \lambda \\
   e_i & \text{otherwise.}  
  \end{cases} 
\end{align} 
We apply this change of basis first with $b=f_1$, and apply another change of basis with $b=f_2$, and so on. Note that since $f_i$'s pairwise commute we obtain the desired commutation relations.
}

Such a basis will be called a {\it conjugate basis}. Note that given $\set{e_i}_{i=1}^{n-m}$ the conjugate basis is unique.

{The JW decomposition of type $m$} associated to the choices $I=\Span{z_1,\cdots,z_{n-m}}$, $I'=\Span{x_1,\cdots,x_{n-m}}$, and $W=\Span{a_1,\cdots,a_{2m}}$, where the JW basis is obtained from  $\set{z_{n-m+i},x_{n-m+i}}_{i=1}^{m}$ {using Eq.~(\ref{eq:symplectic_to_jw})}, {is called the canonical JW decomposition.}
Note that $W$ can be identified by $E_m$ once the indices of the basis vectors are shifted down by $n-m$. 

\Lem{\label{lem:JW for W}
{Let $E=I\oplus W\oplus I'$ be a JW decomposition of type $m$ and $\set{a_i}_{i=1}^{2m}$ be a JW basis for $W\cong E_m$.}
For $b\in W$, {write $\aA_b=\set{a_1,\cdots,a_{2t}}\subset W$ and define} $\set{\hat{a}_i}_{i=1}^{2t}$ {by} 
\[
\hat{a}_1 = b \;\;\text{ and }\;\; \hat{a}_i = a_{i},\quad i = 2,\cdots, 2t.
\]
{Let $\set{\hat e_i,\hat f_i}_{i=1}^{2t}$ be the symplectic basis of $W$ associated to $\set{\hat{a}_i}_{i=1}^{2t}$ given by Eq.~(\ref{eq:jw_to_symplectic}).}
Then, there exists a JW decomposition for $W$ of type {$m-t$:}
\[
W = J \oplus \bar W \oplus J'
\]
such that
\begin{itemize}
\item $J=\Span{\hat e_1,\cdots,\hat e_t}$,
\item $J'=\Span{\hat f_1,\cdots,\hat f_t}$,
\item {$\set{b + a_{2t+1},\cdots,b + a_{2m}}$ is a JW basis of {$\bar W \cong E_{m-t}$}.}
\end{itemize}
Moreover, $J\oplus J'=\Span{a_1,\cdots,a_{2t}}$.
}



\Proof{%
  Notice that $\set{\hat{a}_i}_{i=1}^{2t}$ is an anti-commuting set and $a_{1} = \sum_{i=1}^{2t}\hat{a}_i$ since $\sum_{a \in \aA_b}a = b$. 
{Explicitly,  $\hat e_{i}$ and $\hat f_{i}$ are defined as follows:}
\begin{eqnarray}\label{eq:hat e}
e_{1} = {\hat{a}}_{1}\quad \text{and} \quad e_{i} = {\hat{a}}_{2i-1} + \sum_{k=1}^{2i-2} {\hat{a}}_{k},\quad i = 2,\cdots, t,
\end{eqnarray}
{and}
\begin{eqnarray}\label{eq:hat f}
f_{1} = {\hat{a}}_{2}\quad \text{and} \quad f_{i} =  {\hat{a}}_{2i} + \sum_{k=1}^{2i-2} {\hat{a}}_{k},\quad i = 2,\cdots, t.
\end{eqnarray}  
  So, $J$ and $J'$ are conjugate isotropics and $S = J \oplus J' = \Span{\hat{a}_1, \cdots, \hat{a}_{2t}} = \Span{a_1, \cdots, a_{2t}}$ is a symplectic subspace
  with symplectic basis $\set{\hat{e}_i, \hat{f}_i}_{i=1}^t$. 
  Since $\set{a_i}_{i=1}^{2m}$ is a basis of $W$, we have
  $W = S \oplus \Span{a_{2t+1}, \cdots, a_{2m}}$. Since $b \in S$ we have $W = S \oplus \bar{W}$ where $\bar{W} = \Span{b + a_{2t+1}, \cdots, b + a_{2m}}$.
  Lastly, notice that $\set{b+a_{2t+1},\cdots,b+a_{2m}}$ is pairwise anti-commuting since $\cC_b = \set{a_{2t+1},\cdots,a_{2m+1}}$.
  Hence, $\set{b+a_{2t+1},\cdots,b+a_{2m}}$ is a JW basis for $\bar{W}$.
}
} 
  

{Given an isotropic subspace $I = \Span{e_{1},\cdots,e_{n-m}}$ a JW decomposition (Definition~\ref{def:JW-decomposition}) can be obtained by the following (efficient) procedure: To begin, consider a linear map $L:\mathbb{Z}_{2}^{2n}\to \mathbb{Z}_{2}^{n-m}$ such that $\ker(L) = I^{\perp}$. The corresponding $(n-m)\times 2n$ matrix over $\ZZ_{2}$ will also be denoted {by} $L$.%
~We now outline an efficient procedure for obtaining a {JW} decomposition.%
~To construct a basis for $I^{\perp} = I\oplus W$:%


\begin{itemize}
\item Use Gaussian elimination (GE) to solve the set of homogeneous linear equations $Lx = 0_{n-m}$. This yields a set of vectors $\{v_{1},\cdots,v_{n+m}\}$.
\item Use GE once again to identify $2m$ vectors $\{v_{i_{1}},\cdots,v_{i_{2m}}\}$ from $\{e_{1},\cdots,e_{n-m}\}$ that are linearly independent.
\end{itemize}
We now have that $I^{\perp} = \Span{e_{1},\cdots,e_{n-m},v_{i_{1}},\cdots,v_{i_{2m}}}$. {Due to the rank-nullity theorem we have that $E_{n} = \ker(L)\oplus \im{(L^{T})}$ wherein we identify the $(n-m)$-dimensional subspace $I^{\prime} = \im{(L^{T})}$ as the basis conjugate to $I$. In practice we can construct a basis for $I'$ as follows:}

\begin{itemize}
\item Use GE to determine the linearly independent columns $\{\bar f_{1},\cdots, \bar f_{n-m}\}$ of $L$,
\item obtain a basis $\{f_{1},\cdots,f_{n-m}\}\subset I^{\prime}$ according to $f_{i} = L^{T}\bar f_{i}$ {since $L\lvert_{I^{\prime}}$ is injective}.
\end{itemize}

As described in {Eq.~(\ref{eq:jw_to_symplectic}),}
 we can obtain a JW {set} for $W$ by first identifying {a} symplectic basis $W = \Span{ e'_{i},\cdots, e'_{m}, f'_{1},\cdots, f'_{m}}$, { then {constructing $a_i$'s for $i=1,\cdots,2m$ and setting $a_{2m+1}=\sum_{i=1}^{2m} a_i$} (Lemma~\ref{lem:maximal_anticommutings_are_odd}).
This can be achieved efficiently using symplectic Gram-Schmidt orthogonalization (SGSO) \cite{silva2001lectures,koenig2014efficiently}. Moreover, the procedure outlined in Lemma~\ref{lem:conjugate basis} is a special case of SGSO in which the decomposition $W^{\perp} = I\oplus I^{\prime}$ is given, and thus is also efficient. The protocol as a whole is, in turn, efficient as well since the most expensive subroutines are GE and SGSO, both of which are polynomial time algorithms.}

\subsection{{Classification of CNC sets}}

{The original classification} of CNC sets in~\cite{raussendorf2020phase} (see also~\cite{kirby2019contextuality}) was based on graph-theoretic methods. In this section, we provide an alternative proof using symplectic techniques, centered on our notion of JW decomposition (Definition~\ref{def:JW-decomposition}).

{We begin with the observation that anti-commuting sets together with isotropic subspaces can be used to build CNC sets. We begin with the observation that if $\gamma$ is a value assignment on an isotropic subspace $J$ and ${J}$ is contained in a larger isotropic subspace $J'$, then the value assignment can be extended $\tilde \gamma:J'\to \ZZ_2$.} 
Let $I$ be an isotropic subspace of $E_n$ and $\set{a_1, \cdots, a_N}$ be an anti-commuting set in $I^\perp$. Then,
\begin{equation}\label{eq:decomposition}
  \Omega = \bigcup_{k=1}^N I_k
\end{equation}
where $I_k = \Span{a_k, I}$ is a CNC set. {To see this, first observe that the closedness property under commutation is clear. To construct a value assignment start from a value assignment $\gamma_0:I\to \ZZ_2$ and extend it to each $I_k$ to obtain value assignments $\gamma_k:I_k\to \ZZ_2$. Then the family $\set{\gamma_k}_{k=1}^{N}$ of value assignments which match on $I$ defines a value assignment $\gamma:\Omega\to \ZZ_2$.}

{The classification theorem says that the converse of this observation is also true.} 
Let us start by considering a closed subset $\Omega$. We write

\begin{equation}\label{eq:center}
I(\Omega) = \set{a\in \Omega \mid [a,b]=0\;\forall b\in \Omega},
\end{equation}
and call this {isotropic} subspace the {\it center} of $\Omega$. {Assume that $\operatorname{dim}(I(\Omega)) = n - m$. Then we can write $I(\Omega)^\perp = I(\Omega) \oplus W(\Omega)$ where $W(\Omega)$ is a $2m$-dimensional symplectic subspace in {a} {JW} decomposition.
}
When there is no confusion we simply write $I$ and $W$ for the subspaces associated to the closed subset. 
Using the closedness of $\Omega$, we can write it as a union of isotropic subspaces {of the form $\Span{a,I}$ where {$a\in \Omega\setminus I$}. Note {that} the difference is empty {when $\Omega = I$}.}
Also, $\Span{a, I} = \Span{\bar{a}, I}$ for every $a \in \Omega \setminus I$ where $\bar{a}$ is the projection of a onto $W$.
Therefore we can write $\Omega$ as in Eq.~(\ref{eq:decomposition}) where
$a_k \in W \setminus \set{0}$ for all $k$ and $\Span{a_k, I} \neq \Span{a_l, I}$ for $k \neq l$. 
Proving that $\{a_1, \cdots, a_N\}$ forms an anti-commuting set requires additional work, {which we defer to Section~\ref{sec:proof of cnc classification}. The proofs of the following two results can be found there.}

\Thm{\label{thm:classification CNC}
Let $\Omega \subset E_n$ be a CNC subset. {Then there exists an anti-commuting set $\set{a_1, \cdots, a_N} \subset W(\Omega)$ such that}
\begin{equation}\label{eq:Omega N}
\Omega = \bigcup_{k=1}^N \Span{a_k,I{(\Omega)}}
\end{equation} 
where $1 \le N \le 2m + 1$ and $\dim(I{(\Omega)})=n-m$ {for some $1 \le m \le n$}. 
}

\Cor{{The maximal CNC sets are precisely those that have a decomposition as in Eq.~(\ref{eq:Omega N}) with $N=2m+1$.}\label{cor:maximality}}


{
{We call a maximal CNC set $\Omega$ of type $(n,m)$ if $I(\Omega)$ has dimension $n-m$ and $N=2{m}+1$ in Eq.~(\ref{eq:Omega N}).}
{For each such $\Omega$ there exists a symplectic isomorphism $\alpha: E_n \to E_n$ such that
\begin{equation}\label{eq:symp_map_between_n_m_cncs}
  \alpha(I) = \Span{z_1, \cdots, z_{n-m}}, \quad \alpha(I') = \Span{x_1, \cdots, x_{n-m}}, \quad \alpha(\set{a_i}_{i=1}^{2m+1}) = \set{\hat{a}_i}_{i=1}^{2m+1}
\end{equation}     
where $\set{\hat{a}_i}_{i=1}^{2m+1}$ is the JW set 
associated with the symplectic basis $\set{x_{n-m+i}, z_{n-m+i}}_{i=1}^{m}$.
Hence $\Sp_{2n}(\ZZ_2)$ acts transitively on the set of maximal CNC sets of type $(n,m)$. 

The surjective group homomorphism ${\varphi}:\Cl_{n}\to \Sp_{2n}(\ZZ_{2})$ implies there exists a Clifford unitary \(U_\alpha \in \Cl_n\) such that \(\varphi(U_\alpha) = \alpha\).
Moreover, the phases associated with the conjugation action of \(U_\alpha\) on the Pauli operators \(\{T_{v_i}\}_{i=1}^{2n}\) where \(\{v_i\}_{i=1}^{2n}\) is a (not necessarily symplectic) basis of \(E_n\) can be chosen freely\footnote{{To see this}, notice that there is a local Clifford $V \in Cl_1$ such that
\begin{align*}
  V (X) = X, \quad V(Y) = -Y, \quad V(Z) = Z
\end{align*}
and if we apply it to the $n$-th qubit of the $\set{T_{\hat{a}_i}}_{i=1}^{2m+1}$, it only flips the phase of the $T_{\hat{a}_{2m+1}}$. Combining $V$ with $U_{\alpha}$ gives the desired result.}.
Therefore for each maximal CNC set $\Omega$ of type $(n,m)$, there exists a Clifford unitary $U_\alpha$ such that 
\begin{equation}\label{eq:cnc_clifford_stab_factor_decomposition}
  A_\Omega^\gamma = U_\alpha^\dagger \left(\Pi_{\Span{z_1, \cdots, z_{n-m}}}^0 \otimes A_{\tilde{\Omega}}^{0}\right) U_\alpha
\end{equation} 
where $\tilde{\Omega}$ is a maximal CNC of type $(m,m)$ and $0$ refers to the value assignment that assigns the value $0$ to each element in the corresponding set. {Conversely, a $(n{-}m)$-qubit stabilizer state and an $m$-qubit CNC operator can be composed---i.e., tensored as in Eq.~(\ref{eq:cnc_clifford_stab_factor_decomposition})---to
 produce an $n$-qubit CNC operator.}

This lets us to compute the number of phase point operators; see \cite[Theorem 5]{zurel2020hiddenthesis}.} 
The number $N_{n,m}$ of {maximal} CNC {sets} of type $(n,m)$ is given by

\[
N_{n,m} = N_{n-m} \cdot N_{m,m} 
\]
where
\begin{itemize}
\item $N_{n-m}$ is the number of isotropic subspaces of dimension $n-m$:
\[
N_{n-m} = \prod_{k=1}^{n-m} \frac{4^{n-k+1}-1}{2^k-1},
\]
\item $N_{m,m}$ is the number of {maximal} CNC sets of type $(m,m)$, i.e., the number of distinct sets of JW elements:
\[
N_{m,m} = \frac{|\Sp_{2m}(\ZZ_2)|}{(2m+1)!},
\]
where $|\Sp_{2m}(\ZZ_2)|=2^{m^2}\prod_{j=1}^m {(4^j-1)}$. 
\end{itemize}
{Additionally, the number ${M}_{n,m}$ of maximal CNC operators of type $(n,m)$ is given by
\[
  {M}_{n,m} = N_{n,m} \cdot V_{n,m}
\]
where} $V_{n,m}$ is the number of value assignments on a CNC set of type $(n,m)$:
\[
V_{n,m} = 2^{n+m+1}.
\]
 
\subsection{Enlarging to maximal CNC}
\label{sec:Enlarging to maximal CNC}

Given a CNC set 
\[
\Omega  = \bigcup_{k=1}^N I_k,
\]
where $I_k=\Span{a_k,I}$, $\dim(I)=n-m$, and {$N < 2m+1$}, 
there are two ways to enlarge to a maximal CNC {set}: 
\begin{enumerate}
\item Assume that $N=2l+1$ for some {$1 \le l < m$}. For an isotropic subspace $J\subset I^\perp$ of dimension $t=m-l$ {satisfying $I \cap J = 0$}, we can construct a maximal CNC set $\tilde \Omega$ defined by
\[
\tilde \Omega  = \bigcup_{k=1}^{2l+1} \tilde I_k,
\]
where $\tilde I_k=\Span{a_k, \tilde I}$ and {$\tilde I=I \oplus J$}. 
After choosing a basis $J=\set{e_{n-m+1},\cdots,e_{n-m+t}}$ we can also construct the value assignments $\tilde \gamma:\tilde \Omega\to \ZZ_2$. Such value assignments are of the following form: For $s\in \ZZ_2^t$, we define the value assignment
\begin{equation}\label{eq:value assignment extension}
\gamma\ast s(a) =\begin{cases}
\gamma(a) & a\in \Omega \\
    s_i      & a=\hat e_{n-m+i}.
\end{cases}
\end{equation}
Note that on other Pauli elements the values are automatically determined.

\item Next, we provide an enlargement where the center is fixed but the {size of the anti-commuting set is} increased. Complete the set $\set{a_1,\cdots,a_l}$ to a {JW set}. It is only possible when $\set{a_1,\cdots,a_l}$ is not maximal and we can do this by following Algorithms 1 and 2 in \cite{sarkar2019sets}. Let $t=2(m-l)$. Then, for $s\in \ZZ_2^t$, we define  
\begin{equation}\label{eq:value assignment extension2}
\gamma\ast s(a) =\begin{cases}
\gamma(a) & a\in \Omega \\
    s_i      & a=a_{2l+1+i}.
\end{cases}
\end{equation}
 
\end{enumerate}

\section{The update rules}
\label{sec:the update rules}

In this section, we describe the update rules for CNC operators, as defined in Eq.~(\ref{eq:CNC operator}). Under a Pauli measurement $T_b$, a CNC operator $A_\Omega^\gamma$ is updated to a probabilistic mixture of CNC operators.
For the update rules {of $A_\Omega^\gamma$}, we need the following notation:
\begin{itemize}
\item Given $b\in {\Omega}$, we write ${\gamma + }[b,\cdot]:\Omega\to \ZZ_2$ for the value assignment defined by $a\mapsto {\gamma(a) + }[b,a]$. 
\item {Given $b \in E_n\setminus \Omega$}, we write $\Omega(b)$ for the closed set defined by
\[
\Omega(b) = \Span{b}+ \Span{b}^\perp \cap \Omega
\]  
{and for $r \in \mathbb{Z}_2$ }we write $\gamma \ast r : \Omega(b)\to \ZZ_2$ for the value assignment defined by
\[
\gamma \ast r (a) =
\begin{cases}
\gamma(a) & a\in \Span{b}^\perp \cap \Omega \\
\gamma(b+a)+r+ \beta(b,a) & a\in b+\Span{b}^\perp \cap \Omega.
\end{cases}
\]
Alternatively, $\gamma\ast r$ can be described as the unique value assignment determined by $\gamma\ast r(b)=r$ and $\gamma\ast r(a)=\gamma(a)$ for $a\in \Span{b}^\perp \cap \Omega$.
\end{itemize}

If outcome $r\in \ZZ_2$ is observed the update of a CNC operator under measurement of a Pauli operator $T_b$ is given by the following formula \cite{raussendorf2020phase}:
\begin{equation}\label{eq:updates}
\Pi_b^r A_\Omega^\gamma \Pi_b^r = \begin{cases}
\displaystyle\delta_{r,\gamma(b)}\frac{A_{\Omega}^{\gamma}+A_{\Omega}^{\gamma + [b,\cdot]}}{2} & b\in \Omega \\[6pt]
\displaystyle\frac{1}{2}A_{\Omega(b)}^{\gamma \ast r} & b\not\in \Omega.
\end{cases}
\end{equation}
So, there are two cases for the update rule determined by whether $b\in \Omega$, or not.

This formula generalizes the well-known update of a stabilizer projector under a Pauli measurement.
 Let $I$ be an isotropic subspace of $E_n$ and $s$ be a value assignment on $I$. Then
  \[ \Pi_b^r\Pi_I^s\Pi_b^r = 
  \begin{cases} 
     \delta_{r,s(b)} \Pi_I^s &  b \in I \\[6pt]
      \displaystyle\frac{|\Span{b}^\perp \cap I|}{|I|} \Pi_{I(b)}^{s\,\!*\,\!r}&   b \notin I.
  \end{cases}
  \]

Next, we refine the update rule given in Eq.~(\ref{eq:updates}) as a preparation for the tableau representation in Section \ref{sec:tableau representation}.
We have the following containments of subspaces
\[
I \subset I_k \subset \Omega \subset I^\perp \subset E_n
\]
where $I_k=\Span{a_k,I}$.
The refined update rules will break into four cases:
\begin{itemize}
\item {\bf Case I:} $b\in I$,
\item {\bf Case II:} $b\in I_k\setminus I$,
\item {\bf Case III:} $b\in I^\perp \setminus \Omega$,
\item {\bf Case IV:} $b\notin I^\perp$.
\end{itemize}
{For our simulation, we consider only maximal CNCs. So, assume that $\Omega$ is a maximal CNC set of type $(n,m)$. Then, Theorem~\ref{thm:classification CNC} implies that $\set{a_1, \cdots, a_{2m+1}}$ is a JW set in $W(\Omega) = \Span{a_1, \cdots, a_{2m}}$.
Using this, we can distinguish the first three cases} by the number of JW elements that anti-commute with $b$. {Note that $|\aA_b|$ is even by Corollary~\ref{cor:commuting vs anticommuting}.}

\Lem{\label{lem:three cases}  
For $b\in I^\perp$, let us define the integer
\[
N_b = {\frac{|\aA_b|}{2}}.
\]
Then, we have three cases:
\begin{enumerate}
\item $b\in I$ if and only if $N_b = 0$.

\item  $b\in I_k{\setminus}I$ if and only if $N_b=m$. 

\item  $b\in I^\perp{\setminus}\Omega$ if and only if  $1\leq N_b\leq m-1$.  
\end{enumerate}
}
\Proof{
Part (1): {Note that $\Span{\Omega} = I^\perp$. Then we can see that }$b\in I$ if and only if all JW elements commute with $b$, i.e., $|\aA_b|=0$. 
Part (2): {It is clear that if $b \in I_k\setminus I$ then $N_b = m$. Now, Let $N_b = m$ and $\bar{b}$ denote the projection of $b$ onto $W$. Then $\bar{b} = a_k$ for some $k$ by {Lemma~\ref{lem:JW basis}} which implies $b \in I_k\setminus I$.}
Part (3) holds in the remaining case, i.e., $b$ does not belong to $\Omega$ if and only if  $|\aA_b|\neq 0, 2m$.
}

Cases I and II fall under the case $b\in \Omega$ of Eq.~(\ref{eq:updates}) and the CNC set remains the same under update, {although the value assignment may not}. Cases III and IV require more work. Given a maximal CNC set $\Omega$, we consider the associated JW decomposition $(I,I',W)$ where 
\begin{align*}
I&=\Span{e_1,\cdots,e_{n-m}}\\
I'&=\Span{f_1,\cdots,f_{n-m}}\\
W&=\Span{a_1,\cdots,a_{2m}}.
\end{align*} 
The bases $\set{e_i}_{i=1}^{n-m}$ and $\set{f_i}_{i=1}^{n-m}$ are chosen to be conjugate (Lemma \ref{lem:conjugate basis}).

\paragraph{Case III} 
Let $b\in I^\perp \setminus \Omega$ and $\Omega$ be a maximal CNC set.
We begin by reindexing the JW elements so that $\aA_b=\set{a_1,\cdots,a_{2t}}$, where $1\leq t\leq m-1$. 
Then, observe that $I(b)=\Span{e_1,\cdots,e_{n-m}, b}$.
When we compute $\Omega(b)$ we immediately find that 
\begin{equation}\label{eq:Omega(b) CaseIII}
\Omega(b) = \bigcup_{k=2t+1}^{2m+1} I_k(b)
\end{equation}
where $I_k(b) = \Span{a_k,I(b)}$. A conjugate basis can be given so that $I(b)'=\set{f_1,\cdots,f_{n-m},a}$ for some $a\in \aA_b$.
{The CNC set $\Omega(b)$ is maximal if and only if $t =1$. If $t > 1$, then w}e construct a maximal CNC by enlarging its center (see Section \ref{sec:Enlarging to maximal CNC}).
  The maximal CNC set  
\begin{equation}\label{eq:tilde Omega(b) CaseIII}
\tilde \Omega(b) = \bigcup_{k=2t+1}^{2m+1} \tilde I_k(b)
\end{equation}
containing $\Omega(b)$ and the corresponding conjugate bases
are described as follows:  
\begin{itemize}
\item $\tilde I(b)' = \Span{\bar{f}_1, \cdots, \bar{f}_{n-m+t}}$ where
\begin{equation}\label{eq:caseIII f}
\bar f_i = \begin{cases}
\hat f_{i-(n-m)} &  n-m+1 \leq i \leq n-m+t \\
   f_i & \text{otherwise}  
  \end{cases} 
\end{equation}
and
{%
\begin{equation}\label{CaseIII hat f}
\hat f_{i} = \begin{cases}
{a_{2}} & i=1 \\
a_{2i} + \sum_{k=1}^{i-1} (\hat e_{k}+\hat f_{k}) & 2\leq i\leq t,
\end{cases}
\end{equation}
}%

\item $\tilde I(b) = \Span{\bar{e}_1, \cdots, \bar{e}_{n-m+t}}$ where
\begin{equation}\label{eq:caseIII e}
\bar e_i = \begin{cases}
\hat e_{i-(n-m)} &  n-m+1 \leq i \leq n-m+t \\
   e_i & \text{otherwise}  
  \end{cases} 
\end{equation}
and%
{%
\begin{equation}\label{CaseIII hat e}
\hat e_{i} = \begin{cases}
{\bar{b}}  & i=1 \\
a_{2i-1} + \sum_{k=1}^{i-1} (\hat e_{k}+\hat f_{k}) & 2\leq i\leq t,
\end{cases}
\end{equation}
{where $\bar{b}$ is the projection of $b$ onto $W$ found using $\bar{b} = \sum_{a \in \aA_b} a =  \sum_{a' \in \cC_b} a'$ from Corollary~\ref{cor:commuting vs anticommuting} and the fact that $[b, a_i] = 0$ if and only if $[\bar{b}, a_i]$ for all $i$.}%
}%

\item $\tilde I_k(b) = \Span{\bar a_k,\tilde I(b)}$ where
\begin{equation}\label{eq:caseIII a}
\bar a_k =  {\bar{b}} + a_k,  
\end{equation}
 and $k=2t+1,\cdots,2m+1$.
\end{itemize}
{This construction comes directly from Lemma \ref{lem:JW for W} {and the corresponding equations Eq.~(\ref{eq:hat e}) and Eq.~(\ref{eq:hat f})}.}%


Given the description of $\tilde \Omega(b)$ in Eq.~(\ref{eq:tilde Omega(b) CaseIII}) and its subset $\Omega(b)$ in Eq.~(\ref{eq:Omega(b) CaseIII}), a value assignment $\gamma'$ on $\Omega(b)$ can be extended to  $\tilde \Omega(b)$ by choosing values $s_i \in \ZZ_2$ for the basis vectors $\hat e_i$ where $2\leq i\leq t$. This gives the value assignment $\gamma'\ast s$ described by the construction in Eq.~(\ref{eq:value assignment extension}). 
Then, we can write the Case III update rule as follows:
\begin{equation}\label{eq:Case III maximal version}
\Pi_b^r A_{\Omega}^\gamma \Pi_b^r = 
\frac{1}{2} \frac{\sum_{s\in \ZZ_2^{t-1}} A_{\tilde \Omega(b)}^{(\gamma\ast r)\ast s} }{2^{t-1}}.
\end{equation}

\paragraph{Case IV} Now, we turn to the final update rule. Let $b\notin I^\perp$ and $\Omega$ be a maximal CNC set. 
The maximal CNC set  
\[
 \Omega(b) = \bigcup_{k=1}^{2m+1}  I_k(b)
\] 
with the corresponding conjugate bases are described as follows: Let $\lambda$ be the minimum of the set $\{1\leq i \leq n-m \mid [e_i, b] = 1 \}$ of indices.
\begin{itemize}
\item $I(b)' = \Span{\bar{f}_1, \cdots, \bar{f}_{n-m}}$ where
\begin{equation}\label{eq:caseIV f} 
 \bar{f}_{i} = \begin{cases}
  e_\lambda & i=\lambda\\
    f_{i} + e_{\lambda} & \text{if } [f_i, b] = 1\\
     f_{i} & \text{otherwise.}\\
  \end{cases}
\end{equation}
\item $I(b) = \Span{\bar{e}_1, \cdots, \bar{e}_{n-m}}$ where
\begin{equation}\label{eq:caseIV e}
\bar e_i = \begin{cases}
b & i=\lambda \\
  e_i + e_\lambda & \text{if } [e_i,b]=1 \text{ and } i \neq \lambda \\
   e_i & \text{otherwise.}  
  \end{cases} 
\end{equation}
\item $I_k(b) = \Span{\bar a_k, I(b)}$ where 
\begin{equation}\label{eq:caseIV a}
\bar{a}_k = \begin{cases} 
  a_k + e_\lambda & \text{if } [a_k, b] = 1\\
a_k & \text{if } [a_k, b] = 0 \end{cases}
\end{equation}
and  $k=1,\cdots,2m+1$.
\end{itemize}
This construction relies on the basis transformation described in Eq.~(\ref{eq:modified basis}).

{We conclude this section with a remark on Case III. If the CNC set after the measurement is not maximal, we apply the first enlargement procedure described in Section \ref{sec:Enlarging to maximal CNC}. As we will demonstrate next, this is the only feasible method for achieving a maximal CNC, since the alternative approach---enlarging the JW elements---is not possible.}
Let $\Omega \subset E_n$ be a maximal CNC set with $\on{dim}(I) = n-m$ and JW elements $\set{a_1, \cdots, a_{2m+1}}$. 
Then, after measuring $b \in I^\perp \setminus \Omega$ (reindex JW elements such that $\mathcal{A}_b = \set{a_1, \cdots, a_{2t}}$) we have the CNC set 
\[
    \Omega(b) = \bigcup_{a \in \mathcal{C}_b} \Span{I(b), a} = \bigcup_{k = 2t+1}^{2m+1} \Span{I(b), a_k}
\]
where $I(b) = \Span{I, b}$ is the center of $\Omega(b)$. 
{Let} $\tilde{W}$ {be a symplectic subspace} such that $I(b)^\perp = I(b) \oplus \tilde{W}$.
Then, we claim that any choice of an anticommuting set $\set{\hat{a}_1, \cdots, \hat{a}_{2(m-t)+1}} \subset \tilde{W} \cap \Omega$ such that
\[
    \Omega(b) = \bigcup_{k=1}^{2(m-t)+1} \Span{I(b), \hat{a}_k}
\]
is maximal. To see this, first note that 
\[ 
    \sum_{a \in \mathcal{C}_b} a = \sum_{k=2t+1}^{2m+1} a_k = b \in I(b).
\]
Additionally, without loss of generality we have that $\hat{a}_k = {c_k} + a_{2t+k}$ for some $c_k \in I(b)$ {where $k =1, \cdots, 2(m-t)+1$}.
This gives
\[ 
    \sum_{k = 1}^{2(m-t)+1} \hat{a}_k \in I(b)
\]
and $\tilde{W} \cap I(b) = {0}$ implies
\[ 
    \sum_{k = 1}^{2(m-t)+1} \hat{a}_k = 0.
\]
{Therefore}, {by Proposition \ref{pro:maximal_anticommuting}} $\set{\hat{a}_1, \cdots, \hat{a}_{2(m-t)+1}}$ is a maximal anticommuting set in $\tilde{W}$.

\section{The tableau representation}
\label{sec:tableau representation}

In this section, we introduce our phase space tableau, {{also} referred to as the CNC tableau,} and describe its update rules. These rules are implemented in \cite{CNCSim}; see Fig.~(\ref{fig:fig}).

\subsection{The {phase space} tableau}
 
{As shown in   Theorem~\ref{thm:classification CNC}}, a maximal CNC set {$\Omega$} has the form
\[
\Omega = \bigcup_{k=1}^{2m+1} I_{k}
\]  
where $I{ = I(\Omega)}$ is the center{,} $I_k=\Span{a_k,I}$, {and $\set{a_1, \cdots, a_{2m+1}}$ is a $JW$ set in $W = W(\Omega)$}.
{Here we chose} a JW decomposition $E_n=I\oplus W\oplus I'$ associated to the CNC set {$\Omega$} where $I'$ is a conjugate isotropic.  
Our goal is to construct a tableau representation for $\Omega$ and a value assignment $\gamma$ on this CNC set given by a matrix whose rows encode the conjugate basis vectors and the JW elements.

For an element $a\in E_n$ expressed as $a=\sum_{i=1}^{n} (\alpha_i {x}_i + \beta_i {z}_i)$ we write
\[
R(a) = {(\alpha_1, \cdots ,\alpha_n, \beta_1, \cdots ,\beta_n )}
\] 
for the row vector consisting of the coefficients. For a matrix $M$ we write $M_{i,\ast}$ for its $i$-th row vector.

\Def{\label{def:tableau}
Given a JW decomposition $E_n=I\oplus W\oplus I'$ of type $m$ together with a basis $\set{e_1,\cdots, e_{n-m}}$ for $I$, a conjugate basis $\set{f_1,\cdots,f_{n-m}}$ for $I'$, and an associated JW {set} $\set{a_1,\cdots,a_{2m+1}}$ for $W$, we define a {$(2n)\times(2n+1)$} matrix 
\[
\tT_{i,\ast} = \begin{cases}
R(f_i) & 1\leq i\leq n-m \\
R(e_{i-(n-m)}) & n-m+1 \leq i\leq 2(n-m) \\
R(a_{i-2(n-m)}) & 2(n-m)+1 \leq i\leq 2n+1
\end{cases} 
\]
called a {\it tableau representation {of type $m$}} for the JW decomposition {of type $m$. We denote the $i$-th row of the tableau by $R_i = \tT_{i,\ast}$.}
}

We separate the matrix $\tT=\tT_\Omega$ into three parts:
\begin{itemize}
\item Destabilizer part $\tT|_D $ defined by
\[
(\tT|_D)_{i,\ast} = R_i
\]
where $1\leq i\leq n-m$.

\item Stabilizer part $\tT|_S $ defined by
\[
(\tT|_S)_{i,\ast} = R_{n-m+i}
\]
where $1\leq i\leq n-m$.

\item Jordan-Wigner part $\tT|_{JW} $ defined by
\[
(\tT|_{JW})_{i,\ast} = R_{2(n-m)+i}
\]
where $1\leq i\leq 2m+1$.

\end{itemize}

Next, we extend the tableau to include a value assignment $\gamma:\Omega\to \ZZ_2$. The values of $\gamma$ on the basis of $I$ and the JW elements are recorded in the last column. 
{In destabilizer part, we extend the value assignment with zero values for convenience, even though these values are irrelevant to the simulation algorithm.}

\Def{\label{def:tableau CNC} 
A \emph{tableau representation for a {maximal} CNC pair} {$(\Omega,\gamma)$} is the matrix $\tT_\Omega^\gamma$ obtained from  a tableau representation of the associated JW decomposition by appending a column consisting of the value assignments to the basis elements, that is, $(\tT_\Omega^\gamma)_{i,{j}}=\tT_{i,{j}}$ for {$1 \leq j \leq 2n$,} $1\leq i\leq 2n+1$, and 
\[
(\tT_\Omega^\gamma)_{i,2n+1} =\begin{cases} 
0  &  1\leq i\leq n-m\\
\gamma(e_{i-(n-m)}) & n-m+1 \leq i\leq 2(n-m)\\
\gamma(a_{i-2(n-m)}) &  2(n-m)+1 \leq i\leq 2n{+1}.\\
\end{cases}
\]
{If the maximal CNC has type $(n,m)$, we say that the tableau $\tT_\Omega^\gamma$ is of type $m$.}
}
{A simple algorithm to construct the tableau for a given maximal CNC operator $A_\Omega^\gamma$ is provided in Algorithm~\ref{alg:tableau_from_CNC}.}
For notational simplicity, we write {$\tilde \tT=\tT_\Omega^\gamma$.} The rows of this matrix are denoted by {$\tilde R_i$} and they have the form
\[
{\tilde R_i = (R_i , s_i)}
\]
where $R_i$ is the $i$-th row of $\tT$ and $s_i$ is the value assigned by $\gamma$ to the corresponding Pauli element. Similarly, we partition $\tilde\tT$ into destabilizer, stabilizer, and JW parts:%
\[%
{\tilde \tT|_D,\;\;\; \tilde \tT|_S, \;\;\; \tilde \tT|_{JW}.}
\]
 
\Ex{
{\rm
{Let $I=\Span{z_1,\cdots,z_{n-m}}$, $I'=\set{x_1,\cdots,x_{n-m}}$, and $\set{{a_{1},\cdots,a_{2m+1}}}$ denote the corresponding JW {set}.}
We wish to depict the corresponding tableau. For convenience let us group together even and odd elements $a_{2k}$ and ${a}_{2k-1}$, respectively. We have
\[
\left [\begin{array}{c c|c c c}
I_{n- m} & ~~ & ~~ & ~~ & ~~\\
~~ & ~~ & I_{n- m} & ~~ & ~~\\
~~ & A & ~~ & B & ~~\\
~~ & B & ~~ & A & ~~\\
~~ & 1_{m} & ~~ &  1_{m} &~~\\

\end{array}\right ]
\]
Here $ 1_{m}$ is a row vector of $m$ ones, $I_{n- m}$ is an $(n- m)\times (n-m)$ identity matrix, and empty spaces represent matrices of all zeros of the appropriate dimensions. $A$ and $B$ are $m\times m$ lower triangular and strictly lower triangular matrices, respectively, satisfying
\begin{eqnarray}
A_{ij} =
\begin{cases}%
1\quad i\leq j\\
0\quad \text{otherwise}
\end{cases}%
\quad\text{and}\quad
B_{ij} =
\begin{cases}%
1\quad i < j\\
0\quad \text{otherwise.}
\end{cases}\notag%
\end{eqnarray}
}}

{
\Ex{\label{ex:composition}
{\rm Consider CNC tableaus $\tilde \tT$ and $\tilde \tT^{\,\prime}$ of type $(n,m)$ and $(n',0)$, respectively. These tableaus are of the form
\[
\begin{bmatrix} D & ~ \\S & \gamma_{S}\\W & \gamma_{W} \end{bmatrix}\quad \text{and}\quad %
\begin{bmatrix} D' & ~\\S' & \gamma_{S'}\\~ & ~ \end{bmatrix},
\]
where the $\gamma$'s indicate phase bits.%
~Following 
{Eq.~(\ref{eq:cnc_clifford_stab_factor_decomposition})}
their composition{, i.e., tensor product,} is given by 
{the}
tableau
\[
\begin{bmatrix} D & ~ & ~\\ ~ & D' & ~\\S & ~ & \gamma_{S}\\ ~ & S' & \gamma_{S'} \\ W & ~ & \gamma_{W}\end{bmatrix}.
\]
}
}
}




\subsection{Updating the {phase space} tableau}

Just as in the stabilizer tableau approach, a CNC tableau can be updated either through the deterministic action of Clifford unitaries or the probabilistic action of Pauli measurements. The rules for updating 
under Clifford action is exactly the same as in the stabilizer approach. The $n$-qubit Clifford group $\Cl_{n}$ is generated by the Hadamard and phase gates $H_{i}$ and $S_{j}$ that act nontrivially only on a single qubit, where
\begin{eqnarray}
H = \frac{1}{\sqrt{2}}\begin{bmatrix} 1 & 1\\ 1 & -1\end{bmatrix}\quad \text{and}\quad
S = \begin{bmatrix} 1 & 0\\ 0 & i\end{bmatrix},\notag
\end{eqnarray}
together with the $2$-qubit CNOT entangling gate $CX_{ct} = \ket{0}\bra{0}_{c}\otimes I_{t}+ \ket{1}\bra{1}_{t}\otimes X_{t}$, where $c$ and $ t$ are the control and target qubits. 
Recall that there is a surjective group homomorphism ${\varphi}:\Cl_{n}\to \Sp_{2n}(\ZZ_{2})$. By surjectivity the images of the Clifford generators by $h_{i} := {\varphi}(H_{i})$, $s_{j} := {\varphi}(S_{j})$, and $cx_{ct} := {\varphi}(CX_{ct})$ themselves generate the symplectic group $\Sp_{2n}(\ZZ_{2})$.%

The rules for updating 
under the action of $h_{i}$, $s_{j}$, or $cx_{ct}$ (together with keeping track of the phase bit) are the same as those 
given in \cite{aaronson2004improved}. Supposing a circuit consists of $L$ Clifford gates, these updates can be performed in $O(nL)$ time.

{Next, we turn to the update rules for the Pauli measurements. Given $a\in E_n$ and $s\in \ZZ_2$, consider the row vector
\[
\tilde R(a) = ( R(a) ,s ),
\] 
obtained by appending to $R(a)$ the value $s$ as the last column entry. 
Given another element $b\in E_n$ and value $r\in \ZZ_2$, we can compute the commutator and the $\beta$ function (see Eq.~(\ref{eq:beta})) as follows:
\[
[a,b] = R(a) \omega R(b)^T,\quad \text{where}\quad%
\omega = \begin{pmatrix} 0 & I\\ I & 0 \end{pmatrix},
\]
and%
{%
\begin{eqnarray}
\beta(a,b)%
&=& \frac{1}{2}\left [%
\sum_{i=1}^n \left( R(a)_i R(a)_{i+n} + R(b)_i R(b)_{i+n} \right )\right. \notag\\%
&+&\left .2\sum_{i=1}^n \left (R(a)_{i+n} R(b)_{i}\right ) -  \sum_{i=1}^n (R(a)\oplus R(b))_{i}( R(a)\oplus R(b))_{i+n}  )\right ].\notag
\end{eqnarray}
}

We define {the} \emph{{row} product operation} between two rows:
\begin{equation}\label{eq:row product}
\tilde R(a) \ast  \tilde R(b) = (R(a)+R(b) , s+r+\beta(a,b) ).
\end{equation}
{Notice that the computation of $\beta(a,b)$ requires $O(n)$ steps so that the row product itself also has $O(n)$ complexity in total.}
}

{Our first task is to identify which of the measurement cases from Section~\ref{sec:the update rules} apply.}

{%
\begin{lem}\label{lem:case-check-complexity}
For a maximal CNC set $\Omega$ and {$b \in E_n$}, checking which of the cases (i) $b\in I$, (ii) $b\in I_{k}/I$, (iii), $b\in I^{\perp}/\Omega$, or (iv) $b\notin I^{\perp}$ occurs can be performed in $O(n^{2})$ time. 
\end{lem}
}

\begin{proof} 
To determine which update rule to apply for {$b \in E_n$} we must compute the symplectic {form} of $R(b)$ with each row of $\tT_{\Omega}$. We will mainly use Lemma \ref{lem:three cases} to distinguish the cases.
It suffices to perform the matrix multiplication {$\tT_{\Omega}(b) = \tT_{\Omega}\,\omega\,R(b)^{T}$}, which can be done in $O(n^{2})$ time. Once we have computed the symplectic {form}s, 
we observe that a Pauli element $b\in \Omega$ must satisfy either (i) $b\in I$ or (ii) $b \in I_k\setminus I$ for some unique $k=1,\cdots,2m+1$. The former occurs if and only if $b$ commutes with all rows $n-m{+1}$ to $2n+1$, whereas the latter occurs if and only if $b$ commutes with rows $n- m+1$ to $2(n- m)$ plus one additional row in $2(n-m) +1$ to $2n+1$.%
~On the other hand, if $b {\notin \Omega}$ then 
we have either (iii) $b\in I^{\perp}\setminus \Omega$ or (iv) $b\notin I^{\perp}$. The latter happens if and only if $b$ anti-commutes with any element in rows $n- m+1$ to $2(n-m)$, whereas the former occurs if $b$ commutes with {all rows $n-m$ to $2(n-m)$} but anti-commutes with $2t$ rows in $2(n-m)+1$ to $2n+1$.
\end{proof}

{We are now ready to describe the update rules for the tableau representation using the results of Section \ref{sec:the update rules}.  Consider a maximal CNC set $\Omega$ and suppose we perform a measurement of a Pauli operator{\footnote{{If we measure $-T_b$ we flip the outcomes $0 \leftrightarrow 1$.}}} $T_{b}$ with outcome $r_b$. 
{Let $\bar{b}$ be the projection of $b$ onto $W$. Then we define $\tilde{R}(b) = (R(b), r_b)$ and $\tilde{R}(\bar{b}) = (R(\bar{b}), r_b)$.}
\begin{itemize}
\item \textbf{Case I}. $b\in I$: Deterministic outcome $r_{b}=\gamma(b)$ and deterministic tableau update $\tT_{\Omega}^{\,\gamma}\mapsto \tT_{\Omega}^{\,\gamma}$, i.e., tableau remains untouched.

\item \textbf{Case II}. $b\in I_{k}\setminus I$:  Deterministic outcome $r_{b}=\gamma(b)$ and probabilistic tableau update $\tT_{\Omega}^{\,\gamma}\mapsto \tT_{\Omega}^{\,\gamma}$ or $\tT_{\Omega}^{\,\gamma}\mapsto \tT_{\Omega}^{\,\gamma+[b,\cdot]}$ with equal probability.

\item \textbf{Case {III}}. $b\in I^{\perp}\setminus \Omega$: Sample $r_b\in \ZZ_2$ uniformly as the outcome. Sample ${s = (s_{2},\cdots,s_{t})}\in \ZZ_2^{t-1}$ uniformly at random to update  the tableau $\tT_{\Omega}^{\,\gamma}\mapsto \tT_{\tilde \Omega(b)}^{\,(\gamma \ast r_b)\ast s}$ as follows: First, order the rows of $\tT_{\Omega}^\gamma|_{JW}$ in  such a way that the new rows $\tilde R'_{1},\cdots,\tilde R'_{2m+1}$ satisfy
\[
R'_i \omega R(b)^T = \begin{cases}
1 & 1 \leq i \leq 2t \\
0 & \text{otherwise.}
\end{cases}
\]
{%
We make the assignment%
$$ R'_{1} \gets \sum_{i={2t+1}}^{{2m+1}} R'_{i},$$
which computes $R(\bar b )$ {using Corollary~\ref{cor:commuting vs anticommuting}} that we use in place of $R(b)$.
}%
Then, the new tableau {has type $m-t$} and is given by {using Lemma~\ref{lem:JW for W} as}
\begin{align*}
(\tT_{\tilde\Omega(b)}^{(\gamma \ast r_b)\ast s}|_D)_{i,\ast} &=\begin{cases}
( R'_{2},0) & i=n-m+1\\ 
{( R'_{2k} + \sum_{j=1}^{2k-2}R'_{j},0)} &  i = n-m+k,~~ k = 2,\cdots,t \\
(\tT_\Omega^\gamma|_D)_{i,\ast} & \text{otherwise}
\end{cases}\\
(\tT_{\tilde\Omega(b)}^{(\gamma \ast r_b)\ast s}|_S)_{i,\ast} &=\begin{cases}
(R'_{1},r_{b}) & i=n-m+1 \\
{( R'_{2k-1} + \sum_{j=1}^{2k-2}R'_{j},s_{k})} &  i = n-m+k,~~ k = 2,\cdots,t \\
(\tT_\Omega^\gamma|_S)_{i,\ast} & \text{otherwise}
\end{cases}\\
(\tT_{\tilde\Omega(b)}^{(\gamma \ast r_b)\ast s}|_{JW})_{i,\ast} &=
\tilde R'_{2t+i} \ast \tilde R({\bar{b}})  \;\; \text{ where } 1\leq i\leq 2(m-t){+1}.
\end{align*}

\item \textbf{Case IV}. $b\notin I^{\perp}$: Sample $r_b\in \ZZ_2$ uniformly as the outcome and update  the tableau $\tT_{\Omega}^{\,\gamma}\mapsto \tT_{\Omega(b)}^{\,\gamma\ast r_{b}}$ by as follows: Let $\lambda$ be the minimum of the indices satisfying
\[
(\tT_\Omega^\gamma|_S)_{i,\ast} \omega R(b){^T} =1.
\] 
Then,
\[
(\tT_{\Omega(b)}^{\gamma\ast r_b})_{i,\ast} = \begin{cases}
\tilde R_{\lambda+n-m} & i=\lambda \\
\tilde R(b) & i=\lambda+n-m \\
\tilde R_i \ast \tilde R_\lambda & R_i \omega R(b)^T =1 \text{ and } i\neq \lambda, \lambda +n-m\\
\tilde R_i & \text{otherwise.}
\end{cases}
\]

\end{itemize}
}
{%
Let us discuss how these updates are implemented in practice.%
}
\begin{itemize}
{%
\item \textbf{Case I}: The outcome $r_{b}$ is deterministic and can be computed directly from ${\tT}_{\Omega}(b)$ using the destabilizer entries in the first $n-m$ rows. Similar to the stabilizer tableau \cite{aaronson2004improved}, we have that%
$$R(b) = \sum_{i=1}^{n-m}[R(b),R_{i}]R_{i+n-m}.$$%
Introduce an additional row of scratch space and compute $\tilde R(b) = \tilde R_{i_{1}+n-m}\ast \tilde R_{i_{2}+n-m}\ast \cdots \ast \tilde R_{i_{K}+n-m} $, where $i_{k}\in \{1,\cdots, n-m\}$ is such that $[R(b),R_{i_{k}}]$ is nonzero. The deterministic outcome $r_{b} = \gamma(b)$ is the phase bit of $\tilde R(b)$.

\item \textbf{Case II}: This follows similarly to Case I with the only difference that%
$$R(b) = R_{j}+\sum_{i=1}^{n-m}[R(b),R_{i}]R_{i+n-m}$$%
for some unique $j\in \{2(n-m)+1,2n+1\}$ such that $[R(b),R_{j}]=0$. From this $\tilde R(b)$ and $r_{b}=\gamma(b)$ can be computed.

\item \textbf{Case III}: This case is somewhat more involved. Detailed pseudocode can be found in Algorithm~\ref{alg:case-iii}. Recall first that $R(b)$ anti-commutes with $2t$ $(1\leq t\leq m-1)$ rows in $\tT|_{JW}$. Letting $K = 2(n-m)$, these $2t$ rows are moved to rows $K+1,\cdots, K+2t$. Next, we make the assignment%
$$ R_{K+1} \gets \sum_{i=K+2t+1}^{2n+1}R_{i}$$%
which computes $R(\bar b)$ {using Corollary~\ref{cor:commuting vs anticommuting}} and assigns it to the $(K+1)$-st row. Randomly choose $r_{b}\in \ZZ_{2}$ as the outcome. The new stabilizer is $\tilde R_{K+1} = (R_{K+1},r_{b})$ and the new destabilizer is $\tilde R_{K+{2}} = (R_{K+2},0)$. For $i=2,\cdots,2t$ we do as follows. Introduce a row of scratch space $R({\text{S}})$ initialized at $R({\text{S}}) = R_{K+1}+R_{K+{2}}$ and randomly choose $s_{i}\in \ZZ_{2}$. We make the reassignments:
\begin{eqnarray}
R_{K+2i-1} &\gets & R_{K+2i-1}+R({\text{S}}), \quad \tilde R_{K+2i-1} = (R_{K+2i-1},s_{i}),  \notag\\
R_{K+2i} &\gets & R_{K+2i}+R({\text{S}}), \quad \quad\tilde R_{K+2i} = (R_{K+2i},0),  \notag\\
R({\text{S}}) &\gets & R({\text{S}}) + R_{K+2i-1} + R_{K+2i} \notag
\end{eqnarray}
which are now the $i$th new stabilizer and destabilizer, respectively. Update .%
~For all $i>K+2t$ we update the rows as
\[\tilde R_{i} \gets \tilde R_{i}\ast \tilde R_{K+1}.
\]
We complete the update by making in-place modifications to the tableau, such that for $i=1,\cdots,t$ the rows $\tilde R_{K+2i}$ are in positions $n-m+1$ to $n-m+t$ and rows $\tilde R_{K+2i{-1}}$ are in positions $K+t+1$ to $K+2t$.
}

{%
\item \textbf{Case IV}: Here $b\notin I^{\perp}$ and thus $R(b)$ anti-commutes with one or more stabilizer rows. The update proceeds in a fashion similar to stabilizer theory. We identify a stabilizer row with least index $\lambda$ such that $[R(b),R_{\lambda}]=1$ then for all $i > \lambda$ such that $[R(b),R_{i}]=1$ we make the reassignment%
\[\tilde R_{i} \gets \tilde R_{i}\ast \tilde R_{\lambda}.
\]
{which is similar to Eq.~(\ref{eq:modified basis}) and t}his guarantees that $R(b)$ anti-commutes only with $R_{\lambda}$. Make a random choice $r_{b}\in \ZZ_{2}$. We now make the updates%
\[%
 R_{\lambda-n+m} \gets R_{\lambda} \quad \text{and}\quad  \tilde R_{\lambda }\gets (R(b),r_{b})
\]
to the destabilizer and stabilizer, respectively. Pseudocode can be found in Algorithm~\ref{alg:case-iv}.%
}


\end{itemize}




\begin{thm}\label{thm:cnc-tableau}
The following statements are valid for the CNC tableau:
\begin{enumerate}
\item A CNC tableau can be encoded using $O(n^{2})$ bits of memory.

\item The update of a CNC tableau under:
\begin{itemize}
\item Clifford unitaries can be performed in $O(n)$ time,
\item Pauli measurements can be performed in {$O(n^{2})$ time}.
\end{itemize}
\end{enumerate}

\end{thm}

\begin{proof}%
{We adopt a notion of complexity in which a basic computational step consists of bitwise multiplication or addition.} The proof of the first statement follows from the fact that our CNC tableau is a ${(}2n+1{)}\times {(}2n+1{)}$ matrix over $\ZZ_{2}$ requiring $4n^{2}+4n+1$ bits, which has a linear memory overhead increase over the stabilizer tableau. The update under Clifford unitaries follows in analogy to the stabilizer case, with an $O(1)$ increase due to the additional row that we keep in memory.%
{%
~Let us turn to the analysis of the measurements. First, by Lemma~\ref{lem:case-check-complexity}, it takes $O(n^{2})$ rather than $O(n)$ time to check which of the measurement updates to perform.%
~For Cases I and II we simply require the procedure for outputting the deterministic outcome, which requires $O(n)$ row product operations and thus $O(n^{2})$ operations in total.%
~{Generating the new stabilizer and destabilizer in {Case III}} rows requires summing $O(t)$ rows, thus $O(nt)$ basic operations and $O(n^{2})$ operations in the worst case. {Updating the commuting JW rows in Case III} requires performing $O(m-t)$ row product operations, which has $O(n^{2})$ complexity in the worst case.%
~Finally, {Case IV} requires us to, in the worst case, update $n+m$ rows that anti-commute with $R(b)$, which requires $O(n(n+m))$ operations. {In all cases, scratch space used is $O(n)$. Hence the overall memory complexity is $O(n^{2})$.}
}


\end{proof}




\subsection{Analysis of phase space update rules}

\noindent We now perform a more in-depth analysis of the performance of the phase space tableau, drawing comparisons to the CHP simulator {\cite{aaronson2004improved}} where relevant.%
~A distinguishing feature of our tableau method is that we allow for measurement of general multi-qubit Pauli observables. 
~Thus the CNC approach has an intrinsic quadratic overhead required to \textit{check} which update rule to implement (Lemma~\ref{lem:case-check-complexity}). CHP requires time only linear in $n$ to check if a computational basis measurement will produce a deterministic or random outcome.


\begin{figure}
    \centering
    \begin{subfigure}{0.75\textwidth}
        \centering
        \includegraphics[width=\textwidth]{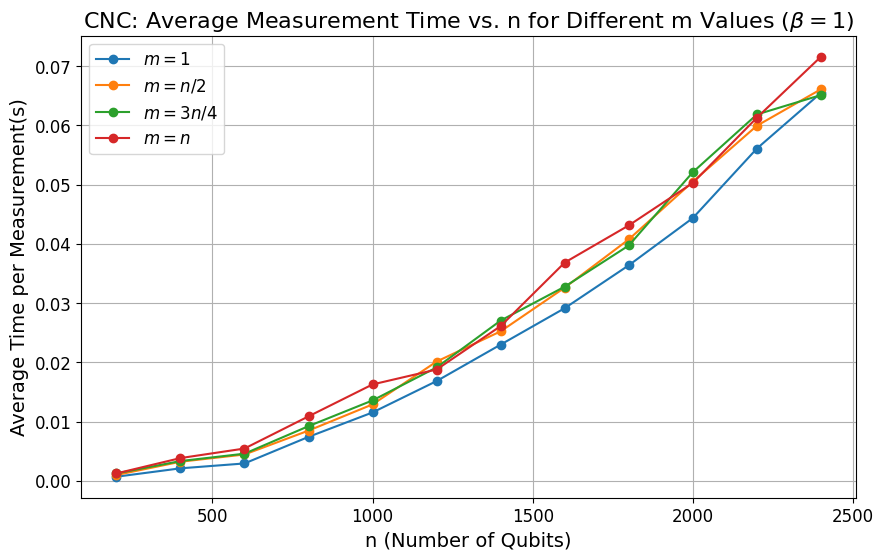}
        \caption{}
        \label{fig:vary-m}
    \end{subfigure}
    
    \vspace{0.5cm} 

    \begin{subfigure}{0.75\textwidth}
        \centering
        \includegraphics[width=\textwidth]{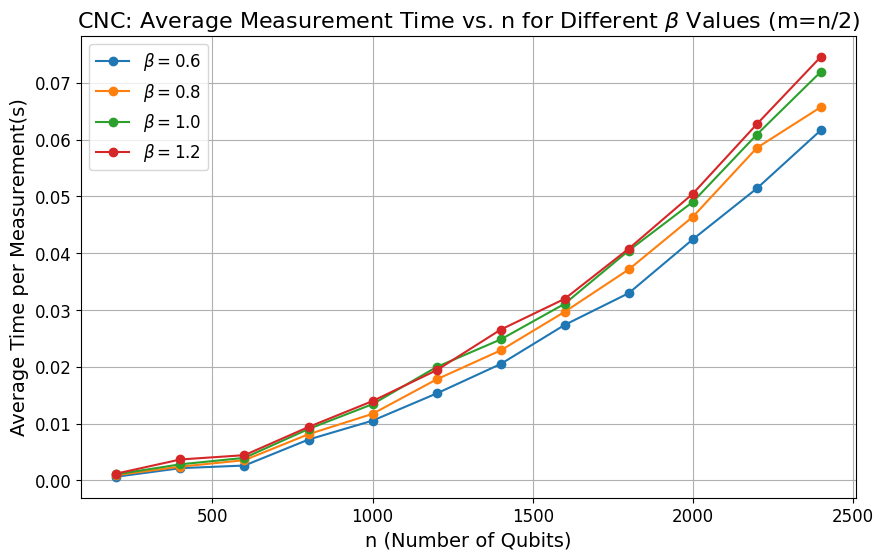}
        \caption{}
        \label{fig:cnc-vary-beta}
    \end{subfigure}
    
    \caption{Figures show the tested measurement time as a function of the qubit count for random Clifford circuits, consisting of $\lfloor\beta n\log_{2} n\rfloor$ gates on $n$ qubits, simulated using our model \cite{CNCSim}. (a) Fixed $\beta = 1$, varying $m$. (b) Varying $\beta$, fixed $m = n/2$. Simulations were run on an Apple M2 MacBook Air with $16$ GB RAM.}
    \label{fig:fig}
\end{figure}

Let us now turn to the update rules.
A novel ingredient in our simulation is the introduction of an additional parameter $m$. There are two extreme cases to consider. When $m=0$ we essentially have the stabilizer tableau and only Case I and Case IV are performed, corresponding to the deterministic and random cases.%
~On the other end of the spectrum, we have $m=n$, in which case only Case II and Case III are performed.%
~{To evaluate how the performance of our simulator scales with respect to $m$, we conduct the following computational experiment, following the setup in~\cite{aaronson2004improved}:}
\begin{itemize}
\item Fix $n, m$ and $\beta > 0$.
\item Choose $\lfloor\beta n\log_{2} n\rfloor$ Clifford gates $H_{i}$, $S_{j}$, $CX_{ij}$, each with probability $1/3$.
\item Computational basis measurement of each qubit in sequential order.
\end{itemize}
We ran this experiment for a fixed value of $\beta =1$ and with $n$ ranging from $200$ to $2,400$ in increments of $200$. For each $n$ we chose for $m$ the values $0$, $1$, $n/4$, $n/2$, $3n/4$, and $n$. We find that increased $m$ results in increased runtime, likely because with increased $m$ there are more applications of the Case II and Case III update rules, which are more complex as compared to Case I and Case IV, respectively; see {Figure~\ref{fig:vary-m}}.


In \cite{aaronson2004improved} they vary the parameter $\beta$ and observe a ``phase transition" in which the simulation shifts from a roughly linear to quadratic runtime. This is due to the increased number of calls to the row product operation as $\beta$, i.e. the number of Clifford {unitaries}, increases.\footnote{Our row product operation corresponds to the \textbf{rowsum} operation of AG \cite{aaronson2004improved}.} The CNC tableau simulation, on the other hand, is less sensitive to changes in $\beta$; see Figure~\ref{fig:cnc-vary-beta}. This is largely due to the fact that we have an intrinsic quadratic overhead just to check which case to apply. However, the CNC simulation is not entirely insensitive to $\beta$ since in real case scenarios Case I and Case II updates are more efficient than Case III and Case IV updates, and the former occur less frequently as $\beta$ increases.%



When we fix $m=0$ the CNC tableau has the structure of the AG tableau, apart from an additional row that has all zero entries in this case. We compared the performance of our CNC simulator with $m=0$ to a Python implementation of the AG tableau due to Craig Gidney \cite{STABSim}. 
Our results can be found in Figure \ref{fig:chp-comparison}. Numerically we see that our implementation is considerably more performant (note the vertical axis scaling), although the reason for this is somewhat unclear since our Case I and Case IV should parallel the determined and random update rules of stabilizer simulation.

\section{Phase space tableau simulation}
\label{sec:phase space tableau simulation}

For classical simulation of QCM we consider an initial $n$-qubit quantum state expanded using CNC operators as a generally quasi-probabilistic mixture
$$\rho = \sum_{\alpha}W_{\rho}(\alpha)A_{\alpha},$$
where {$\alpha$ runs over}
our phase space {points} which consists of all maximal $(n,m)$ CNC pairs. {This initial state may be completely general}; for example, it could be a pure magic state or its approximation if one is willing to tolerate some error in the computation.

Adaptive stabilizer operations necessary for universal quantum computation can be described using \textit{instruments}; see e.g., \cite{watrous2018theory}. For Hilbert spaces $\hH$ and $\kK$, an instrument $\Phi$ with outcome set $\Sigma$ consists of completely positive linear operators $\Phi^{s}:L(\hH)\to L(\kK)$ for all $s\in \Sigma$ such that $\Phi = \sum_{s\in\Sigma}\Phi^{s}$ is a quantum channel. Sequences of instruments can be composed and satisfy the channel condition
\[\Psi = \sum_{{s} \in \Sigma_{1}\times \cdots {\times} \Sigma_{T}} \left (\Phi_{T}\circ \cdots \circ \Phi_{1}\right )^{{s}} = \Psi_{N}\circ \cdots \circ \Psi_{1}.\]
We consider instruments consisting of Pauli measurements and Clifford gates, optionally with classical feedforward and adaptivity, that can be given a Kraus decomposition
\[
\Phi^{s_{i}}({A}) = V_{s_{i}} \Pi_{b}^{s_{i}} \,{A}\,  \Pi_{b}^{s_{i}} V_{s_{i}}^{\dagger}
\]
where ${A}\in L(\hH)$ is an arbitrary linear operator and $V_{{s_{i}}}$ is a Clifford unitary {for $s_{i}\in \Sigma_{i}$}. Observe that non-adaptive Clifford unitaries can always be absorbed into the definition of $V_{{s_{i}}}$. These instruments are referred to as \textit{completely stabilizer preserving instruments} and the corresponding channels are a subset of simulable \textit{completely stabilizer-preserving channels} \cite{seddon2019quantifying}. An important class of examples includes gadgets that deterministically implement non-Clifford gates by performing adaptive stabilizer operations acting on a magic state input.\footnote{{W}e do not consider the most general notion of a completely stabilizer-preserving channel. An extensive study of completely stabilizer preserving operations and their related monotones was performed in \cite{seddon2019quantifying}.}

\begin{figure}
    \centering
    \includegraphics[width=0.75\textwidth]{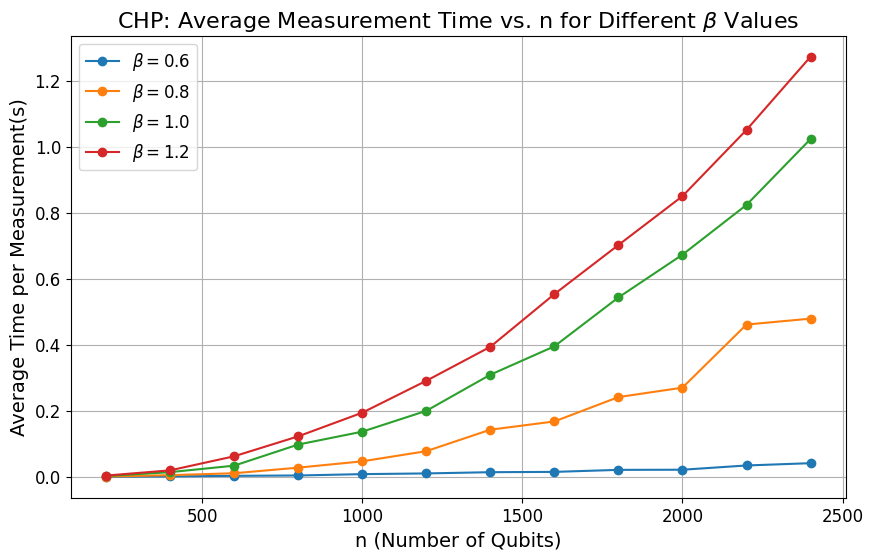}

    \vspace{0.5cm} 

    \includegraphics[width=0.75\textwidth]{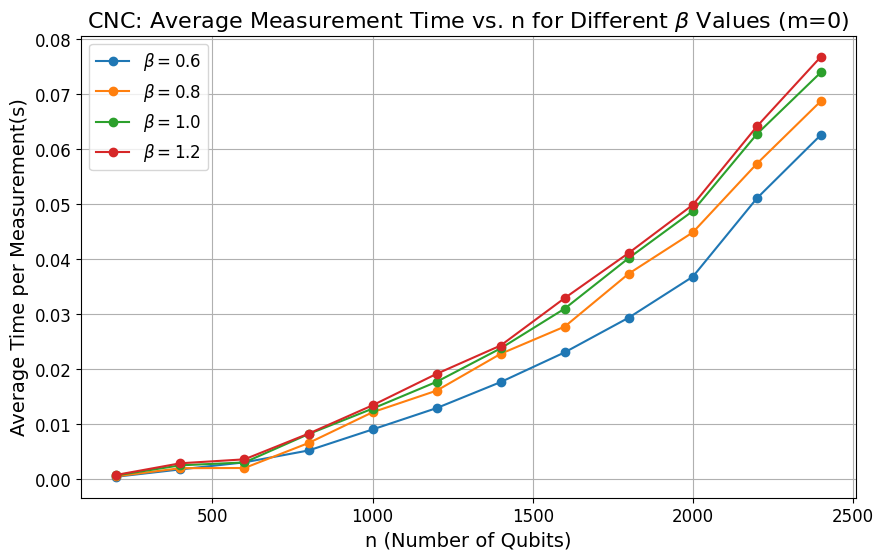} 
    
    \caption{Figures show the tested measurement time as a function of the qubit count for random Clifford circuits, consisting of $\lfloor\beta n\log_{2} n\rfloor$ gates on $n$ qubits, simulated using our model \cite{CNCSim} and CHP implementation of Craig Gidney \cite{STABSim}. Both simulations were run on an Apple M1 Pro MacBook Pro with $16$ GB RAM.}
    \label{fig:chp-comparison}
\end{figure}

\subsection{Classical simulation of QCM}

There are two distinct sampling based classical simulation tasks one can perform depending on whether $W_{\rho}$ represents a proper probability distribution or not. 
~In the former case we perform a \textit{weak simulation} \cite{jozsa2013classical,pashayan2020estimation}, in which the output of the simulation are samples from the true quantum mechanical probability distribution obtained via the Born rule. When negativity is present, however, one cannot directly sample from $W_{\rho}$ and weak simulation is not possible. One instead aims in this case to \textit{estimate} Born rule probabilities using the methods of \cite{pashayan2015estimating}; see also \cite{howard2017application}.

\subsubsection{Weak simulation}

Let us consider a sequence $\Phi_{1},\cdots,\Phi_{T}$ of $n$-qubit stabilizer instruments.%
~Quantum states $\rho$ for which $W_{\rho}$ is a probability distribution lie in $\text{CP}_{n}$, the CNC polytope given by the convex hull of maximal $n$-qubit CNC operators. Weak simulation proceeds by first sampling an initial phase space element $\alpha_{0}$ from the probability distribution $W_{\rho}$. Each stabilizer instrument $\Phi_i$ induces a stochastic map $\alpha_{i-1} \mapsto (s_i, \alpha_i)$, where the pair $(s_i, \alpha_i)$ is sampled from a conditional probability distribution $Q_{\alpha_{i-1}, i}$. {This distribution is determined by
\[
\Phi_{i}(A_{\alpha_{i-1}}) = V_{s_{i}}\Pi_{a_{i}}A_{\alpha_{i-1}}\Pi_{a_{i}} V_{s_{i}}
\]
and Eqs.~(\ref{eq:updates}) and~(\ref{eq:Case III maximal version}), up to a known Clifford unitary. Here $\alpha_{i-1} = (\Omega_{i-1},\gamma_{i-1})$ denotes the phase space element sampled after the application of $\Phi_{i-1}$.}

In practice these stochastic maps are implemented using the phase space tableau so that a stabilizer instrument $\Phi_{i}$ induces an efficient tableau update $\tilde \tT_{\alpha_{i-1}}\mapsto \tilde \tT_{\alpha_{i}}$. Note that the outcome can be retrieved from the tableau. Our procedure for weak simulation using the phase space tableau is given in Algorithm~\ref{alg:weak-sim}.%
~Correctness of our algorithm follows from the correctness of our update rules and Theorem~$3$ of \cite{raussendorf2020phase}. If $W_{\rho}$ can be efficiently sampled from, then efficiency of our tableau update rules (Theorem~\ref{thm:cnc-tableau}) make explicit the efficiency of Algorithm~\ref{alg:weak-sim}.

\subsubsection{Born rule estimation}

Suppose that $\rho$ is an $n$-qubit quantum state that does not lie in $\text{CP}_{n}$ and that $\Psi$ is a stabilizer channel with decomposition into stabilizer instruments $\Phi_{1},\cdots,\Phi_{T}$ so that $\Psi = \sum_{ s} \left (\Phi_{T}\circ \cdots \circ \Phi_{1}\right )^{ s}$, as above. In many cases of interest the initial state is a magic state and $\Psi$ implements a unitary coming from a universal gate set on a subset of the qubits. We wish to estimate the Born rule probability
\begin{eqnarray}
p(x) := \Tr\left (\Pi^{{x}}\Psi(\rho)\right ) = \sum_{\alpha\in V}W_{\rho}(\alpha)\Tr\left (\Pi^{ x}\Psi(A_{\alpha})\right ),%
\label{eq:born-rule}
\end{eqnarray}
where $\Pi^{ x} = \ket{ x}\bra{ x}\otimes I_{n-k}$ is a computational basis measurement (stabilizer projector) on the first $k\leq n$ qubits for a fixed outcome $ x \in \ZZ_{2}^{k}$. This can be rewritten as an expectation value
$$p( x) = \sum_{\alpha\in V}\tilde W_{\rho}(\alpha)\cdot E( x|\alpha)\quad \text{with}\quad E( x|\alpha) =  \text{sign}\left (\alpha\right )\cdot \lVert W_{\rho} \rVert_{1} \cdot   p( x|\alpha),$$%
and where $\sigma(\alpha)$ is the sign of the weight $W_{\rho}(\alpha)$, $\lVert W_{\rho} \rVert_{1}$ is the \textit{negativity} of the quasi-mixture, and $p( x|\alpha) = \Tr\left (\Pi^{ x}\Psi(A_{\alpha})\right )$ is the probability of obtaining $ x$ given an initial phase space element $\alpha$. The average is taken with respect to the normalized probability distribution $\tilde W_{\rho}(\alpha) = |W_{\rho}(\alpha)|/\lVert W_{\rho}\rVert_{1}$.

We pursue a sampling based approach to estimating $p( x)$ using the methods of \cite{pashayan2015estimating,howard2017application}. Fix a number $M$ of samples and consider the estimator
$$\hat p( x) = \frac{1}{M}\sum_{i} \hat E_{i},$$
where $\hat E_{i} = E( x|\alpha_{i})$ and $\alpha_{i}$ is the $i$th initial sample. In many practical cases it is efficient to \textit{estimate} $\hat E_{i}$ since we can efficiently sample from the probability $p( x|\alpha_{i})$ using the phase space tableau, as in weak simulation, and then use a $\text{poly}(n)$ subroutine to produce an estimate up to the desired precision.\footnote{It was shown by Jozsa and Van den Nest \cite{jozsa2013classical} that computing exact Born rule probabilities (i.e., strong simulation) of \textit{adaptive} stabilizer circuits is in the complexity class $\#P$ and thus unlikely to be efficient in general.}

Central to the estimation procedure (Algorithm~\ref{alg:born-rule-estimation}) is that $\hat p( x)$ is an unbiased estimator of the Born rule probability $p( x)$ whose value is bounded by $\pm \lVert W_{\rho} \rVert_{1}$. Hoeffding's inequality \cite{hoeffding1994probability} then implies that to estimate $p( x)$ within an additive error $\epsilon$ with probability $1-\delta$ requires a number of samples given by
$$M \geq  \frac{2}{\epsilon^{2}}\lVert W_{\rho} \rVert_{1}^{2}\log(2/\delta).$$
Each run of the circuit can be performed in $\text{poly}(n)$ time, so the dominant contribution to the complexity of estimation is the number of samples $M$, which scales with $\lVert W_{\rho} \rVert_{1}^{2}$. When the decomposition is optimal the sample complexity scales with $\mathfrak{R}(\rho)^{2}$.



\subsection{Phase space simulation of quantum algorithms}

As a proof of concept we tested our phase space tableau approach \cite{CNCSim}  on two families of deterministic quantum algorithms---one solving the hidden shift problem \cite{rotteler2010quantum}, while the other solves the well-known Deutsch-Jozsa problem \cite{deutsch1992rapid}. In particular, we used quasi-probabilistic methods to estimate the Born rule probability of a fixed deterministic outcome, as described above. 


\subsubsection{Hidden shift algorithm}

In recent years the hidden shift quantum algorithm \cite{rotteler2010quantum} has become widely used to benchmark methods for classically simulating Clifford+$T$ circuits \cite{bravyi2016trading,bravyi2016improved,bravyi2019simulation,pashayan2021fast,kissinger2022simulating}.
{This algorithm was originally introduced by Roetteler in~\cite{rotteler2010quantum} to solve the hidden shift problem and to demonstrate an oracle separation between the complexity classes \textsf{BQP} and \textsf{P}.} A circuit implementation of the hidden shift oracle was given by Bravyi and Gosset in \cite{bravyi2016improved} using the $\{Z,CZ,CCZ\}$ diagonal gate set. An appealing feature of the hidden shift algorithm is that one can control both the number of qubits $n$ and the number of non-Clifford gates $m$ so that there are instances of the problem that are solved by Clifford dominated circuits with large $n$ and relatively small $m$, a task which is impractical using state vector simulation.



For our family of oracles we first fix positive integers $\kappa$ and $\nu\geq 3\kappa$, with $n=2\nu$ and $c=2\kappa$. The parameter $\kappa$ controls the total number of $CCZ$ gates. Consider also a hidden shift string {$ x = (x_{1},\cdots,x_{n})\in \ZZ_{2}^{n}$} as well as a pair of Maiorana-McFarland bent functions $f,f':\ZZ_{2}^{\nu}\to \{\pm 1\}$ such that $f'$ is the Walsh-Hadamard transformation of the shifted version of $f$. We then have that
\[
\ket{x} = U\ket{0},\quad \text{where}\quad U = H^{\otimes n} O_{f'} H^{\otimes n} O_{f} H^{\otimes n}.
\]
The oracles are $n$-qubit diagonal unitary operators that can be constructed explicitly according to
\[
O_{f} = \left (\prod_{i=1}^{\nu}CZ_{i,i+\nu}\right ) \left (O_{g}\otimes I_{\nu}\right )\quad \text{and}\quad O_{f'} = \left (\prod_{i=1}^{\nu}CZ_{i,i+\nu}\right ) \left (I_{\nu} \otimes O_{g}\right ) Z(x).
\]
Here $Z(x) = Z^{x_{1}}\otimes \cdots \otimes Z^{x_{n}}$ is an $n$-qubit unitary and $O_{g}$ is a $\nu$-qubit diagonal unitary corresponding to an arbitrary Boolean function $g:\ZZ_{2}^{\nu}\to \ZZ_{2}$. For our purposes we choose $O_{g}$ to be of the form
\[
O_{g} = \prod_{j=0}^{c-1}CCZ_{j_{1},j_{2},j_{3}}\quad \text{where}\quad j_{k} = k+jc\quad (k=1,2,3).
\]
If $n> 3\nu$ then one can perform an arbitrary unitary on the remaining qubits using diagonal Clifford gates drawn from $\{Z,CZ\}$.

\subsubsection{Deutsch-Jozsa algorithm}
The Deutsch-Jozsa problem \cite{deutsch1992rapid} is among the first examples of a problem whose quantum solution shows an exponential separation with the corresponding classical solution. Here we give a brief overview of the quantum algorithm; for additional details, see \cite{nielsen2010quantum}. We are promised that a Boolean function $f:\ZZ_{2}^{n}\to \ZZ_{2}$ is either constant or balanced and we are given access to an oracle $O_{f}$ which is an $(n+1)$-qubit unitary operator whose action on the computational basis is given by $O_{f}\ket{y,z} = \ket{y, z + f(y)}$ with $y\in\ZZ_{2}^{n}$ and where the sum is taken modulo $2$. Implementing the unitary $U = H^{\otimes n}O_{f}H^{\otimes (n+1)}$ we have a final quantum state given by $\ket{\psi} = U\ket{0, 1}$ wherein $\bra{\psi}(\Pi^{0}\otimes I)\ket{\psi}=0,1$ if $f$ is balanced, constant, respectively.

Our explicit construction of the oracle $O_{f}$ is as follows. Fix some positive integer $c$ and let $n\geq 3c$. We consider Boolean functions $f:\ZZ_{2}^{n}\to \ZZ_{2}$ whose polynomial degree is at most quadratic in the inputs. When $f$ is constant we simply take $f(y)=0$ for all $y\in \ZZ_{2}^{n}$ so that $O_{f}$ is just the identity. For the balanced case we take
\[
f(y) = \sum_{j=0}^{c-1}y_{j_{1}}y_{j_{2}}+y_{j_{3}}+\sum_{k=3c+1}^{n}y_{k}.
\]
The function $f$ is readily seen to be a balanced function {since it is a sum of balanced functions with disjoint domains in $\ZZ_{2}^{n}$}. To implement the quantum oracle recall first that the action of $CX$ and the Toffoli gate $CCX$ on the computational basis is $CX_{c,t}\ket{y_{c},y_{t}}=\ket{y_{c},y_{t}+y_{c}}$ and $CCX_{c_{1},c_{2},t}\ket{y_{c_{1}},y_{c_{2}},y_{t}}=\ket{y_{c_{1}},y_{c_{2}},y_{t}+y_{c_{1}}y_{c_{2}}}$, respectively. We then have that
 the unitary implementing the oracle $O_{f}$ is given by
\[
O_{f}= \prod_{j=0}^{c-1} CCX_{j_{1},j_{2},j_{3}}\,CX_{j_{3},t} ~ \prod_{k=3c+1}^{n} CX_{k,t}.
\]
To our knowledge the Deutsch-Jozsa algorithm has not previously been used to benchmark Clifford+$T$ circuits. One advantage of our Deutsch-Jozsa oracle over that used in the hidden shift algorithm is that ours is more flexible in its use of non-Clifford gates since it can allow for any number of $CCX$ gates while the hidden shift oracle always requires an even number of $CCZ = (I\otimes I\otimes H)\, CCX \, (I\otimes I\otimes H)$ gates.



\subsubsection{Numerical implementation and results}


{We used our publicly available code in \cite{CNCSim} to obtain Born rule estimates for the hidden shift and Deutsch-Jozsa circuits.} For a quantum circuit with $n$ qubits and $m$ $T$ gates our initial state in QCM will be given by
\begin{eqnarray}
\rho = \ket{ 0}\bra{ 0}\otimes \left (\ket{T}\bra{T}\right )^{\otimes m}\quad \text{where}\quad \ket{T}=T\ket{+} = \frac{1}{\sqrt{2}}\left (\ket{0}+e^{i\pi/4}\ket{1}\right )\notag
\end{eqnarray}
and $\ket{0}$ is $n$-qubit all zero state. It suffices, in fact, to decompose the $m$-fold tensor product of $\ket{T}\bra{T}$ as a quasi-mixture of CNC operators since 
the tensor product of a CNC operator with a stabilizer state {(Example \ref{ex:composition})} is once again a CNC operator.

Suppose that a CNC decomposition $W^{(k)}_{C}$ is found for $k$-fold tensor product of $\ket{T}$ 
 and that stabilizer decompositions $W_{S}^{(\ell)}$ and $W_{S}^{(z)}$ are found for 
 $\ell$-fold and $z$-fold tensor products of the $\ket{T}$ state, respectively, such that $m = k+ r\ell + z$, for some non-negative integers $r,z$. We then have that the sample complexity scales as
\[
M\sim\frac{1}{\epsilon^{2}} \lVert W_{C}^{(k)} \rVert_{1}^{2}\cdot \lVert W_{S}^{(\ell)} \rVert_{1}^{2r}\cdot \lVert W_{S}^{(z)} \rVert_{1}^{2}.
\]
For fixed $k$ the asymptotic scaling follows the stabilizer part of the decomposition with the CNC part giving a constant improvement over using stabilizer states alone. For practical implementations we were able to obtain a CNC decomposition for $k=4$ and stabilizer decompositions with $\ell =4$ and $z\leq 3$. Precise numerical values of robustness can be found in Table~\ref{tab:comparison robustness}. {We used these decompositions to estimate Born rule probabilities for the hidden shift and Deutsch-Jozsa algorithms; see Table~\ref{tab:sim-results}.}

\begin{table}[h!]
\centering
\renewcommand{\arraystretch}{1.3}
\begin{tabular}{c c c c c}\hline\hline
~~~&%
$\ket{T}$	& $\ket{T}^{\otimes 2}$ 	& $\ket{T}^{\otimes 3}$ 	& $\ket{T}^{\otimes 4}$ \\ \hline
$\mathfrak{R}_{\on{C}}$&%
$1.000$              & $1.000$               & $1.283$               & $\leq 2.172$          \\ \hline
$\mathfrak{R}_{\on{S}}$&%
$1.414$              & $1.748$               & $2.219$               & $2.863$               \\ \hline\hline
\end{tabular}
\renewcommand{\arraystretch}{1}
\caption{Robustness values for the CNC and stabilizer phase spaces with $n\leq 4$ qubits. CNC robustness for $n=4$ is not exact and was obtained using a subset of CNC operators in the optimization problem. For best known upper-bounds on the robustness of magic, see \cite{heinrich2019robustness}.}
\label{tab:comparison robustness}
\end{table}

\begin{table}[h!]
  \centering
  {
  \begin{subtable}[t]{\textwidth}
    \centering
    \renewcommand{\arraystretch}{1.3}
    \begin{tabular}{c c c c c c c c}
      \toprule\toprule
      $n$ & $m$ & $\epsilon$ & $\delta$ & $\lVert W \rVert_{1}$ & $M$ & $\tau$ (s) & $\hat p(x)$ \\
      \midrule
      $6$   & $14$ & $0.1$  & $0.1$ & $31.1$       & $579{,}626$   & $1{,}284$   & $1.010$ \\
      $106$ & $14$ & $0.1$  & $0.1$ & $31.1$       & $579{,}626$   & $19{,}022$  & $1.000$ \\
      $6$   & $14$ & $0.01$ & $0.1$ & $31.1$       & $57{,}962{,}672$ & $127{,}688$ & $0.997$ \\
      \bottomrule\bottomrule
    \end{tabular}
    \renewcommand{\arraystretch}{1}
  \end{subtable}
  
\vspace{2 em}  
  
  \begin{subtable}[t]{\textwidth}
    \centering
    \renewcommand{\arraystretch}{1.3}
    \begin{tabular}{c c c c c c c c}
      \toprule\toprule
      $n$   & $m$   & $\epsilon$ & $\delta$ & $\lVert W \rVert_{1}$ & $M$        & $\tau$ (s) & $\hat p(x)$ \\
      \midrule
      $4$   & $7$   & $0.1$  & $0.1$ & $4.82$      & $13{,}914$    & $13$       & $0.000$ \\
      $104$ & $7$   & $0.1$  & $0.1$ & $4.82$      & $13{,}914$    & $373$      & $0.000$ \\
      $4$   & $7$   & $0.01$ & $0.1$ & $4.82$      & $1{,}391{,}418$ & $1{,}314$   & $0.000$ \\
      $104$ & $7$   & $0.01$ & $0.1$ & $4.82$      & $1{,}391{,}418$ & $36{,}817$  & $0.000$ \\
      $7$   & $14$  & $0.1$  & $0.1$ & $31.1$      & $579{,}626$   & $1{,}162$   & $0.000$ \\
      $107$ & $14$  & $0.1$  & $0.1$ & $31.1$      & $579{,}626$   & $18{,}199$  & $0.000$ \\
      $7$   & $14$  & $0.01$ & $0.1$ & $31.1$      & $57{,}962{,}672$ & $36{,}816$  & $0.000$ \\
      \bottomrule\bottomrule
    \end{tabular}
    \renewcommand{\arraystretch}{1}
  \end{subtable}
}
  \caption{%
  {%
    Simulation results for (top) hidden shift algorithm with $x=(1,\cdots,1)\in \ZZ_{2}^{n}$ and (bottom)  Deutsch-Jozsa algorithm with $x=(0,\cdots,0)\in \ZZ_{2}^{n}$. The run time $\tau$ is measured in seconds. See \cite{CNCSim} for implementation details.
    }
  }
  \label{tab:sim-results}
\end{table}





\section{Conclusion}

{In this paper, we introduce a novel tableau representation for the phase space simulation algorithm based on closed non-contextual (CNC) operators. These operators can be viewed as a generalization of stabilizer states. The resulting simulation algorithm efficiently simulates a broader class of quantum circuits beyond stabilizer circuits. The complexity of our simulator scales similarly to the well-known Gottesman–Knill simulation. {We have implemented our simulator in~\cite{CNCSim} and benchmarked it on basic quantum algorithms, including the hidden shift and Deutsch--Jozsa algorithms.}

{The phase space simulator is part of a more general class of simulators based on p}olytope-theoretic methods. {These methods} can be utilized to provide a systematic approach to the problem of identifying quantum resources responsible for quantum advantage. The $\Lambda$ polytope and the corresponding simulator introduced in \cite{zurel2020hidden} is the first examples of a probabilistic representation of quantum theory on a finite state space. Our phase space points are also vertices of the $\Lambda$ polytope, specifying a region of efficiently simulatable quantum states. However, discovering more vertices of $\Lambda$ can extend this region. Therefore, the systematic and refined analysis of the symplectic structure of the qubit phase space presented in our work serves as an important step in identifying additional combinatorial structures that can be used for efficiently describing and updating more $\Lambda$ vertices. On the other hand, exploring computational models beyond QCM, such as measurement-based quantum computation \cite{raussendorf2001one}, introduces a new family of polytopes that provide a probabilistic representation of quantum theory; see \cite{okay2024classical} for the local Pauli simulator. The tableau method introduced in this work may be extended to incorporate newly discovered types of vertices, ultimately expanding the region of efficiently simulatable quantum states and encompassing a broader class of quantum circuits. {As more vertices are discovered, the simulation polytope expands, leading to a decrease in the corresponding robustness measure. This, in turn, improves efficiency when the algorithm is initialized with a quasi-probabilistic representation of the resource state.}

\bibliography{bib}
\bibliographystyle{ieeetr}

\appendix
{

\section{Proof of CNC classification}
\label{sec:proof of cnc classification}

  \Lem{ 
    Let $p_1, p_2, p_3, p_4 \in E_n$ satisfying
    \begin{equation}\label{eq:mermin_relations}
      \begin{aligned}
            &[p_1, p_2] = [p_1, p_3] = [p_2, p_4] = [p_3, p_4] = 0\\
            &[p_1, p_4] = [p_2, p_3] = 1
      \end{aligned}
    \end{equation}
    Then there is no CNC set that contains $\set{p_1, p_2, p_3, p_4}$.
  }

  \Proof{
    Let $\Omega$ be a CNC set that contains $\set{p_1, p_2, p_3, p_4}$ and $\gamma: \Omega \to \mathbb{Z}_2$ be a value assignment on $\Omega$.
    By closedness, we know that $p_1 + p_2, p_1 + p_3, p_2 + p_4, p_3 + p_4 \in \Omega$. Let $\gamma(p_i) = \alpha_i$ for all $i \in \set{1,2,3,4}$ and $\alpha = \sum_{i=1}^4 \alpha_i$.
    Then, 
    \begin{align*}
      (-1)^\alpha T_{p_1}T_{p_2}T_{p_3}T_{p_4} &= (-1)^{\alpha}((-1)^{\beta(p_1, p_2)}T_{p_1 + p_2})((-1)^{\beta(p_3, p_4)}T_{p_3 + p_4})\\
      &= ((-1)^{\alpha_1 + \alpha_2 + \beta(p_1 + p_2)}T_{p_1 + p_2})((-1)^{\alpha_3 + \alpha_4 + \beta(p_3 + p_4)}T_{p_3 + p_4})\\
      &= (-1)^{\gamma(p_1 + p_2) + \gamma(p_3 + p_4)}T_{p_1 + p_2}T_{p_3 + p_4}\\
      &= (-1)^{\gamma(p_1 + p_2) + \gamma(p_3 + p_4) + \beta(p_1 + p_2, p_3 + p_4)}T_{p_1 + p_2 + p_3 + p_4}\\
      &= (-1)^{\gamma(p_1 + p_2 + p_3 + p_4)}T_{p_1 + p_2 + p_3 + p_4}
    \end{align*}
    By similar calculations we also find that 
    \begin{align*}
      (-1)^\alpha T_{p_1}T_{p_3}T_{p_2}T_{p_4} = (-1)^{\gamma(p_1 + p_2 + p_3 + p_4)}T_{p_1 + p_2 + p_3 + p_4}
    \end{align*}
    after cancellation this implies that $T_{p_2}T_{p_3} = T_{p_3}T_{p_2}$ which is a contradiction because we assumed $[p_2, p_3] = 1$. Hence, there is no such CNC set.
    
  }

  \Cor{\label{cor:additional_mermin_square_conditions}
    Let $q_1, q_2, q_3, q_4 \in \Omega$ for some closed subset $\Omega \subset E_n$. If any of the following conditions are satisfied, $\Omega$ is not a CNC.
    \begin{enumerate}[label=(\roman*)]
      \item $[q_1, q_2] = [q_2, q_3] = [q_1, q_4] = 0$ and $[q_1, q_3] = [q_2, q_4] = [q_3, q_4] = 1$
      \item $[q_1, q_2] = [q_2, q_3] = 0$ and $[q_1, q_3] = [q_1, q_4] = [q_2, q_4] = [q_3, q_4] = 1$
    \end{enumerate}
  }
  \Proof{
    $(i)$ Since $\Omega$ is closed, we have $q_1 + q_4, q_2 + q_3 \in \Omega$. Taking $p_1 = q_3, p_2 = q_2 + q_3, p_3 = q_1 + q_4, p_4 = q_4$ satisfies Eq.~(\ref{eq:mermin_relations}).
    $(ii)$ We have $q_1 + q_2, q_2 + q_3 \in \Omega$ by closedness. Taking $p_1 = q_1 + q_2, p_2 = q_2, p_3 = q_4, p_4 = q_2 + q_3$ satisfies Eq.~(\ref{eq:mermin_relations}). 
  }

  \Lem{\label{lem:transitivy_of_commutation}
    Let $\Omega \subset E_n$ be a CNC set and $a_1, a_2, a_3 \in \Omega$. If $[a_1, a_2] = [a_2, a_3] = 0$ and $a_2 \notin I(\Omega)$ then $[a_1, a_3] = 0$.
  }
  \Proof{
    Assume that $[a_1, a_2] = [a_2, a_3] = 0$ and $a_2 \notin I(\Omega)$ but $[a_1, a_3] = 1$. Take $a_4 \in \Omega$. There are 8 possibilities for the tuple $([a_1, a_4], [a_2, a_4], [a_3, a_4])$.
    We need to show that all of them are impossible. We show it in 5 cases.
    \begin{enumerate}[label=(\roman*)]
      \item $[a_2, a_4] = 0$\\
      We have $[a_i, a_2] = 0$ for all $i \in \set{1,3,4}$. Since $a_2 \notin I(\Omega)$, there exists $b \in \Omega$ such that $[b, a_2] = 1$.
      Let $\eta_1 = [a_1, b]$, $\eta_3 = [a_3, b]$. Then taking $p_1 = a_1 + \eta_1 a_2, p_2 = a_2, p_3 = b, p_4 = \eta_3 a_2 + a_3$ satisfies Eq.~(\ref{eq:mermin_relations}).
      \item $([a_1, a_4], [a_2, a_4], [a_3, a_4]) = (0, 1, 0)$\\
      Taking $p_i = a_i$ for all $i \in \set{1,2,3,4}$ satisfies Eq.~(\ref{eq:mermin_relations}).
      \item $([a_1, a_4], [a_2, a_4], [a_3, a_4]) = (0, 1, 1)$\\
      Taking $q_i = a_i$ for all $i \in \set{1,2,3,4}$ satisfies the case $(i)$ of Corollary~\ref{cor:additional_mermin_square_conditions}.
      \item $([a_1, a_4], [a_2, a_4], [a_3, a_4]) = (1, 1, 0)$\\
      By symmetry, $(iii)$ also proves this case.
      \item $([a_1, a_4], [a_2, a_4], [a_3, a_4]) = (1, 1, 1)$\\
      Taking $q_i = a_i$ for all $i \in \set{1,2,3,4}$ satisfies the case $(ii)$ of Corollary~\ref{cor:additional_mermin_square_conditions}.
    \end{enumerate}
    Hence, all of the cases are impossible and we have a contradiction.
  }
}

\begin{proof}[{\bf Proof of Theorem \ref{thm:classification CNC}}] 
    Let us write {$I=I(\Omega)$ and $W=W(\Omega)$}. If $\Omega = I$ then choosing $N=1$ and $a_1=0$ gives the desired decomposition. So, assume that $\Omega \neq I(\Omega)$ {and $N\geq 2$}. 
    Let $\sim$ be {the} relation on $\set{a_1, \cdots, a_N}$ defined as $x \sim y$ if and only if $[x, y] = 0$. Clearly, it is reflexive and symmetric. {By} Lemma~\ref{lem:transitivy_of_commutation}, 
$\sim$ is also transitive.
{Therefore,} $\sim$ is an equivalence relation {on the set} $\set{a_1, \cdots, a_N}$. {This set} is partitioned {into} equivalence classes {under} $\sim$. 
Moreover, {$\Omega\neq I$} implies {that there are} at least two equivalence classes. 
Let $x_1, x_2$ be two distinct elements in the same equivalence class{, i.e., $[x_1,x_2]=0$}.
{Then,} $x_1 + x_2 \neq 0$ {belongs to} $\Omega$ {by closedness and} $[x_1, x_1 + x_2] = [x_2, x_1 + x_2] = 0$ implies $x_1 + x_2$ must lie in the equivalence class containing $x_1$ and $x_2$. 
Since there are at least two equivalence classes, there must exist $y \in \set{a_1, \cdots, a_N} \setminus \set{x_1, x_2, x_1 + x_2}$ such that $[x_1, y] = [x_2, y] = [x_1 + x_2, y] = 1$.
However, this is not possible since $[x_1 + x_2, y] = [x_1, y] + [x_2, y]$. So, there cannot be an equivalence class of $\sim$ with size greater than $1$. Therefore, all equivalence classes of $\sim$ are singletons, which means $\set{a_1, \cdots, a_N}$ is an anti-commuting set.
    Since $W$ is a $2m$-dimensional symplectic vector space and $\set{a_1, \cdots, a_N}$ is an anti-commuting set in $W$ we must have $N \le 2m+1$. 
\end{proof}

\begin{proof}[{\bf Proof of Corollary \ref{cor:maximality}}]
{First assume that $\Omega=I$.}
Let $I = \Span{e_1, \cdots, e_n}$ {and d}efine $J = \Span{e_1, \cdots, e_{n-1}}$. Then there exists a symplectic subspace $W$ {of} dimension $2$ such that $J^\perp = J \oplus W$. Clearly, $\set{e_n}$ is an anti-commuting set in $W$ that is not maximal. 
    Therefore, we can extend it to a JW {set} $\set{e_n, a_1, a_2}$ in $W$ by Lemma~\ref{lem:extending_nonmaximal_anticommuting_to_jw}. Then  
    $\Omega = \Span{e_n, J} \cup \Span{a_1, J} \cup \Span{a_2, J}$ is a CNC {containing $I$.}

{Next, a}ssume that $1 < N < 2m+1$.
    If $\set{a_1, \cdots, a_N}$ is not a maximal anti-commuting set, then we can extend it to a JW set $\set{a_1, \cdots, a_N, a_{N+1}, \cdots, a_{2m + 1}}$ of $W$ by Lemma~\ref{lem:extending_nonmaximal_anticommuting_to_jw}.
Then, we  see that $\tilde{\Omega} = \bigcup_{k=1}^{2m+1} \Span{a_k, I}$ is {a} CNC set that strictly contains $\Omega$. 
So $\Omega$ is not maximal. If $\set{a_1, \cdots, a_N}$ is a maximal anticommuting set, then $N = 2l +1$ for some $1 \le l < m$ and $a_{2l+1} = \sum_{i=1}^{2l}a_i$.
    So, we have $a_{2l+1} \in S = \Span{a_1, \cdots, a_{2l}}$ for some symplectic subspace $S$ of dimension $2l$.
    Since $W$ and $S$ are symplectic subspaces, we can find conjugate isotropics $J, J'$ and write a JW decomposition $W = S \oplus J \oplus J'$.
    We have $J^\perp = S \oplus J$. Define $\tilde{\Omega} = \bigcup_{k=1}^{2l+1} \Span{a_k, \tilde{I}}$ where $\tilde{I} = I \oplus J$. We can see that $\tilde{\Omega}$ is a CNC set that strictly contains $\Omega$. So, $\Omega$ is not maximal. Therefore, maximal CNC sets have $N = 2m + 1$.

Lastly, we show that if $N = 2m+1$ for some $1 \le m \le n$ then $\Omega$ is a maximal CNC. Assume that $\set{a_1, \cdots, a_{2m+1}}$ is an anti-commuting set in $I^\perp$ and $\Omega = \bigcup_{k=1}^{2m+1} \Span{a_k, I}$. Let $\bar{a}_i$ be the projection of $a_i$ onto $W$. Then $\set{\bar{a}_1, \cdots, \bar{a}_{2m+1}}$ is an anti-commuting set in $W$ and $\Omega = \bigcup_{k=1}^{2m+1} \Span{\bar{a}_k, I}$.
    Let $\tilde{\Omega}$ be a maximal CNC set such that $\Omega \subset \tilde{\Omega}$. Denote $I(\tilde{\Omega})$ by $\tilde{I}$ and $W(\tilde{\Omega})$ by $\tilde{W}$. Let $\operatorname{dim}(\tilde{I}) = n - \tilde{m}$ for some $1 \le \tilde{m} \le n$. Then there exists an anti-commuting set $\set{b_1, \cdots, b_{2\tilde{m}+1}}$ in $\tilde{I}^\perp$ such that $\tilde{\Omega} = \bigcup_{k=1}^{2\tilde{m}+1} \Span{b_k, \tilde{I}}$. Since $\set{a_1, \cdots, a_{2m+1}}$ is an anti-commuting set, each $a_i$ must be in different isotropic subspace $\Span{b_k, \tilde{I}}$. Hence we must have $\tilde{m} \ge m$. 
    Also, $|\Omega| = (2m+2)2^{n-m} \le (2\tilde{m}+2)2^{n-\tilde{m}} = |\tilde{\Omega}|$ implies $m \ge \tilde{m}$. Therefore, $m = \tilde{m}$ which implies $\Omega = \tilde{\Omega}$. Hence, $\Omega$ is maximal.
\end{proof}


\section{{Pseudocode for the algorithms}}


\begin{algorithm}[H] 
\textbf{Input:} Current tableau $\bar \tT_{\Omega}^{\gamma}$, Measurement basis $R(b)$, commutation information $\tT_{\Omega}(b)$.

\vspace{0.5 em}

\textbf{Output:} Updated tableau $\bar \tT_{\tilde\Omega(b)}^{(\gamma\ast r_{b})\ast s}$ and outcome $r_{b}$.

\begin{enumerate}
    \item \textbf{Reorganize JW Elements:}
    \begin{itemize}
        \item {Identify indices of commuting and anti-commuting JW elements using $\tT_{\Omega}(b)$.}
    		\item Compute $t \gets \frac{1}{2}|\aA_{b}|$
        \item Swap rows in the tableau to place anti-commuting elements in the first $2t$ rows:
    \end{itemize}
    
{%
\item \textbf{Generate New Stabilizers and Destabilizers:}

    \begin{itemize}
    		\item {Let} $K = 2(n-m)$ and {randomly} {sample the measurement outcome} $r_{b}\in \ZZ_{2}$.
    		\item First Stabilizer and Destabilizer:
    		\begin{itemize}
 		\item \textbf{For} $i = K+2t+1,\cdots,2n+1$:
		\begin{itemize}
		\item[-] $R_{K+1} \gets R_{K+1} + R_{i}$
		\end{itemize}		 		
        \item $\tilde R_{K+1}\gets (R_{K+1},r_{b})$ and $\tilde R_{K+2} \gets (R_{K+2},0)$.
    		\end{itemize}
    		\item {Generate Remaining New Stabilizers and Destabilizers:}
    \begin{itemize}
 		\item Introduce extra row: $R(s) \gets R_{K+1}+R_{K+2}$
        \item \textbf{For} $i = 2, \dots, t$:
        \begin{itemize}
        \item[-] Randomly choose $s_{i}\in \ZZ_{2}$.
        \item[-] Destabilizer: $ R_{K+2i} \gets R_{K+2i} + R(\text{S})$,\quad $\tilde R_{K+2i} \gets (R_{K+2i},0)$,
        \item[-] Stabilizer: $ R_{K+2i{-1}} \gets R_{K+2i{-1}} + R(\text{S})$,\quad $\tilde R_{K+2i{-1}} \gets (R_{K+2i{-1}},s_{i})$,
        
        	\item[-] Update $R(\text{S})\gets R(\text{S})+R_{K+2i-1}+R_{K+2i}$.
        \end{itemize}
    \end{itemize}
    \end{itemize}
    \item \textbf{Update Commuting JW Elements:}
    \begin{itemize}
        \item \textbf{For} $i = K+2t+1,\cdots,2n+1$:
        \begin{itemize}
        \item[-] $\tilde R_{i}\gets  \tilde R_{i}\ast \tilde R_{K+1}$
        \end{itemize}
    \end{itemize}
   	}
    
\item \textbf{Rearrange Stabilizers and Destabilizers:}

    \begin{itemize}
    		\item {Organize rows as follows
        \begin{itemize}
          \item[-] Rows $1, \cdots, n-m$ are the previous destabilizer rows.
          \item[-] Rows $n-m+1, \cdots, n-m+t$ are the new destabilizer rows $K+2i$ ($i=1,\cdots,t$).
          \item[-] Rows $n-m+t+1, \cdots, 2(n-m)+t$ are the previous stabilizer rows.
          \item[-] Rows $2(n-m)+t+1, \cdots, 2\left(n-(m-t)\right)$ are the new stabilizer rows $K+2i-1$ ($i=1,\cdots,t$).
        \end{itemize}
        }
    \end{itemize}
    \item \textbf{Update {type}:}
    \begin{itemize}
        \item Update $m \gets m - t$.
    \end{itemize}
    \item \textbf{Return:} Updated tableau and measurement outcome {$r_b$}.
\end{enumerate}
\caption{\label{alg:case-iii} {Algorithm for simulating measurement outcome and tableau update when $b\in I^{\perp}/\Omega$.}%
}
\end{algorithm}

\newpage


\begin{algorithm}[H] 
\textbf{Input:} Current tableau $\bar \tT_{\Omega}^{\gamma}$, Measurement basis $R(b)$, commutation information $\tT_{\Omega}(b)$.

\vspace{0.5 em}

\textbf{Output:} Updated tableau $\bar \tT_{\tilde\Omega(b)}^{\gamma\ast r_{b}}$ and outcome $r_{b}$.

\begin{enumerate}
    \item \textbf{Find Minimum Anti-commuting Stabilizer:}
    \begin{itemize}
        \item $\lambda \gets \min \{i+n-m\}_{i=1}^{n-m}$ such that $[R_{\lambda},R(b)]=1$.
    \end{itemize}
    \item \textbf{Modify Anti-commuting Rows:}
    \begin{itemize}
        \item \textbf{For} each row $i > \lambda$:
        \begin{itemize}
            \item If $[R_{i},R(b)] = 1$ then:\quad $\tilde R_{i}\gets \tilde R_{i}\ast \tilde R_{\lambda}$
        \end{itemize}
    \end{itemize}
    \item \textbf{Update Stabilizer and Destabilizer:}
    \begin{itemize}
        \item New Destabilizer: $ R_{\lambda - n+m}\gets R_{\lambda}$
        \item New Stabilizer:
        \begin{itemize}
        \item[-] Randomly sample the measurement outcome $r_{b}\in {\ZZ_{2}}$.
        \item[-] $\tilde R_{\lambda}\gets (R(b),r_{b})$.
        \end{itemize}
    \end{itemize}
    \item \textbf{Return:} Updated tableau and measurement outcome $r_{b}$.
\end{enumerate}
\caption{\label{alg:case-iv} {Algorithm for simulating measurement outcome and tableau update when $b\notin I^{\perp}$.}%
}
\end{algorithm}

\newpage
 
{
\begin{algorithm}[H] 
  \caption{\label{alg:tableau_from_CNC} {Algorithm for creating a tableau from given maximal CNC operator $A_\Omega^\gamma$.}}
\textbf{Input:} Maximal CNC operator $A_\Omega^\gamma$.

\vspace{0.5 em}

\textbf{Output:} Tableau $\tT_{\Omega}^{\gamma}$.

\begin{enumerate}
    \item \textbf{Find the Stabilizer:}
    \begin{itemize}
        \item Define $I$ as the center of $\Omega$, i.e., the set of all elements that commute with every element of $\Omega$.
        \item Find a basis $\{e_1,\ldots,e_{n-m}\}$ for $I$.
    \end{itemize}
    \item \textbf{Find the Destabilizers:}
    \begin{itemize}
      \item Extend $\{e_1,\ldots,e_{n-m}\}$ to a symplectic basis $\{e_1, f_1,\ldots, e_n, f_n\}$ of $E_n$ using symplectic Gram-Schmidt orthogonalization.
      \item Then $\{f_1,\ldots,f_{n-m}\}$ is a basis for the destabilizer $I'$.
    \end{itemize}
    \item \textbf{Find JW elements:}
    \begin{itemize}
        \item Define symplectic subspace $W$ as $W = \Span{e_{n-m+1}, f_{n-m+1}, \cdots, e_{n}, f_{n}}$.
        \item Let $P_W$ be the orthogonal projection onto $W$ and set $\Omega' = \Omega \setminus I$.
        \item Initialize $JW := \set{}$.
        \item \textbf{For} $i=1$ \textbf{to} $2m+1$:
        \begin{itemize}
          \item Choose $\hat{a}_i \in \Omega'$ and compute $a_i = P_W(\hat{a}_i)$.
          \item Add $a_i$ to $JW$.
          \item Remove all elements that commute with $a_i$ from $\Omega'$.
        \end{itemize}    
    \end{itemize}
    \item \textbf{Construct the Tableau:}
    \begin{itemize}
      \item Construct tableau $\tT_{\Omega}^{\gamma}$ as follows:
  \[
    \tT_{\Omega}^{\gamma}
  =
  \left(
    \begin{array}{c|c}
      \dashedentry{0.217\linewidth}{f_1} & \gamma(f_1) \\
      \vdots & \vdots \\
      \dashedentry{0.2\linewidth}{f_{n-m}} & \gamma(f_{n-m}) \\
      \dashedentry{0.217\linewidth}{e_1} & \gamma(e_1) \\
      \vdots & \vdots \\
      \dashedentry{0.2\linewidth}{e_{n-m}} & \gamma(e_{n-m}) \\
      \dashedentry{0.217\linewidth}{a_1} & \gamma(a_1) \\
      \vdots & \vdots \\
      \dashedentry{0.193\linewidth}{a_{2m+1}} & \gamma(a_{2m+1})
    \end{array}
    \right)
  \quad
  \begin{array}{l}
  \left.\rule{0pt}{2.6em}\right\}\; \text{Destabilizer}\\[2em]
  \left.\rule{0pt}{2.5em}\right\}\; \text{Stabilizer}\\[1em]
  \left.\rule{0pt}{2.5em}\right\}\; \text{JW elements}\\[2em]
  \end{array}
  \]
    \end{itemize}
    \item \textbf{Return:} Tableau $\tT_{\Omega}^{\gamma}$.
\end{enumerate}
\end{algorithm}
}


\begin{algorithm}
\textbf{Input:} Initial distribution $W_{\rho}\geq 0$. Stabilizer instruments $\Phi_{1},\cdots,\Phi_{N}$.

\textbf{Output:} Set of outcomes $\lL = \{s_{1},\cdots,s_{T}\}$.

\vspace{0.5 em}

\noindent $\bullet$ Initialize empty outcome set $\lL$.\\
\noindent $\bullet$ Sample initial tableau $\tilde \tT_{\alpha_{0}}$ from $W_{\rho}$.

\noindent $\bullet$ For $i=1,\cdots, T$:\\
\quad - For instrument $\Phi_{i}$, tableau update $\tilde \tT_{\alpha_{i-1}}\mapsto \tilde \tT_{\alpha_{i}}$.\\
\quad - Obtain outcome $s_{i}$ from tableau: $\lL\gets \lL\cup \{s_{i}\}$.\\
\noindent $\bullet$ Output $\lL$.\\
\caption{\label{alg:weak-sim} {Algorithm for weak simulation using the phase space tableau.}%
}
\end{algorithm}

\begin{algorithm}
\textbf{Input:} Initial distribution $W_{\rho}\geq 0$. Stabilizer instruments $\Phi_{1},\cdots,\Phi_{N}$. Number of samples $M$.

\textbf{Output:} Estimate of $\pi(\vec x)$.

\vspace{0.5 em}

\noindent $\bullet$ Initialize estimator $\hat \pi \gets 0$.\\

\noindent $\bullet$ For $i=1,\cdots, M$:\\
\quad - Sample $\alpha_{i}$ from $\tilde W_{\rho}$.\\
\quad - Estimate $\hat E_{i}$: $\hat E_{i}\gets E(\vec x|\alpha_{i})$.\\
\quad - $\hat \pi \gets \pi + \hat E_{i}$.\\
\noindent $\bullet$ Output $\pi/M$.
\caption{\label{alg:born-rule-estimation} {Algorithm for producing additive-precision estimates of Born rule probabilities.}%
}
\end{algorithm}



\end{document}